%% file: aoc.tex
\documentclass{birkjour}

\usepackage{amssymb}
\usepackage{bbm}
\usepackage{hyperref}

\input{defbirkjour_article}

\begin{document}
\title[Anderson's Orthogonality Catastrophe]{Anderson's Orthogonality
  Catastrophe for\\ One-dimensional Systems}
\author[K\"uttler]{Heinrich K\"uttler}
\address{Mathematisches Institut\\ LMU M\"unchen \\ Germany}
\email{Heinrich.Kuettler@mathematik.uni-muenchen.de}
\author[Otte]{Peter Otte}
\address{Fakult\"at f\"ur Mathematik \\ Ruhr-Universit\"at Bochum \\ Germany}
\email{peter.otte@rub.de} 
\author[Spitzer]{Wolfgang Spitzer}
\address{Fakult\"at f\"ur Mathematik und Informatik \\ Fernuniversit\"at in Hagen \\ Germany}
\email{wolfgang.spitzer@fernuni-hagen.de}
\thanks{This work was supported by the research network SFB TR 12 -- `Symmetries and
Universality in Mesoscopic Systems' of the German Research Foundation (DFG).
The authors would like to thank Peter M\"uller for numerous stimulating discussions.}

\subjclass{Primary 81Q10, 34L40; Secondary 34L20, 34L25}
\keywords{Many fermion system, transition probability, Anderson integral,
thermodynamic limit}
\date{June 17, 2013}

\begin{abstract}
\input{abstract}

\end{abstract}
\maketitle
\section{Introduction}
\input{introduction}
\section{Representation of the Anderson Integral\label{anderson}}
\input{anderson}
\subsection{Operators with a common spectral gap\label{anderson_gap}}
\input{anderson_spectral_gap}
\subsection{Schr\"odinger-type operators\label{anderson_schroedinger}}
\input{anderson_schroedinger_type}
\section{One-dimensional Schr\"odinger Operators}
\input{schroedinger_operators}
\subsection{Free resolvent\label{f_resolvent}}
\input{free_resolvent}
\subsection{Truncated free resolvent\label{f_truncated}}
\input{free_truncated_resolvent}
\subsection{One-dimensional delta-term, \texorpdfstring{$D(z)$}{D(z)}\label{f_delta}}
\input{free_delta}
\subsection{Perturbed resolvent\label{p_resolvent}}
\input{perturbed_resolvent}
\subsection{Perturbed eigenvalues\label{p_eigenvalues}}
\input{perturbed_eigenvalues}
\section{Delta-estimate\label{delta}}
\input{delta_estimate}
\section{Asymptotics\label{asymptotics}}
\input{asymptotics}
\subsection{Subdominant term\label{asymptotics_s}}
\input{asymptotics_subdominant}
\subsection{Dominant term\label{asymptotics_d}}
\input{asymptotics_dominant}
\subsection{Determinant}
\input{determinant}
\appendix
\section{Estimates\label{estimates}}
\input{appendix_estimates}
%

\bibliographystyle{abbrvurl}
\bibliography{/u/otte/local/Bibliothek/BibTeX/aoc}
\end{document}

%% file: defbirkjour_article.tex
\newtheorem{theorem}{Theorem}[section]   
\newtheorem{proposition}[theorem]{Proposition}
\newtheorem{lemma}[theorem]{Lemma}        
\newtheorem{corollary}[theorem]{Corollary}

\numberwithin{equation}{section}
                                         
\newcommand{\N}{\mathbbm{N}}

\newcommand{\R}{\mathbbm{R}}
\newcommand{\C}{\mathbbm{C}}

\newcommand{\id}{\mathbbm{1}}            

\newcommand{\hilbert}{{\mathcal H}}      

\DeclareMathOperator{\im}{Im}            
\DeclareMathOperator{\re}{Re}            

\DeclareMathOperator{\ran}{ran}          
\DeclareMathOperator{\diag}{diag}        
\DeclareMathOperator{\tr}{tr}            


%% file: abstract.tex
We derive rigorously the leading asymptotics of the so-called Anderson integral
in the thermodynamic limit for one-dimensional, non-relativistic, spin-less Fermi
systems. The coefficient, $\gamma$, of the leading term is computed in terms of
the S-matrix. This implies a lower and an upper bound on the exponent
in Anderson's orthogonality catastrophe, $\tilde CN^{-\tilde\gamma}\leq \mathcal{D}_N\leq
CN^{-\gamma}$ pertaining to the overlap, $\mathcal{D}_N$, of ground states of
non-interacting fermions.

%% file: introduction.tex
In 1967, P.W. Anderson \cite{Anderson1967} studied the transition probability
between the ground state of $N$ free fermions and the ground state of $N$
fermions subject to an exterior (radially symmetric) potential in $\R^3$. Interestingly, he found
that this probability decays like $N^{-\gamma}$ with some explcit $\gamma>0$ 
(in terms of phase shifts of the potential) as
$N\to\infty$. Here, we give a rigorous analysis of this
so-called orthogonality catastrophe for one-dimensional systems.

To begin with, let us briefly sketch the many-particle problem underlying our
considerations. The state space of $N$ fermions is the $N$-fold anti-symmetric
tensor product
\begin{equation*}
  \hilbert^N := \underbrace{\hilbert \wedge \ldots \wedge
    \hilbert}_{N-\text{times}}
\end{equation*}
of some one-particle space $\hilbert$ (e.g. $\hilbert=L^2(\Omega)\otimes\C^s$,
$\Omega\subset\R^d$, $s,d\in\N$)
where a one-particle Hamilton operator
$H:D(H)\to\hilbert$ is defined. Since we assume our particles to not interact
the corresponding operator $H^N$ on $\hilbert^N$ is simply a sum
\begin{equation*}
  H^N := H\wedge\id\wedge\ldots\wedge\id  + \cdots +
  \id\wedge\ldots\wedge\id\wedge H .
\end{equation*}
If $H$ has a discrete spectrum consisting of (simple) eigenvalues $\lambda_1<\lambda_2<\cdots$
with corresponding eigenvectors $\varphi_1,\varphi_2,\ldots$ one can easily construct
the analogous $N$-particle quantities. In particular, the ground state $\varphi^N$
is a Slater determinant and the eigenvalue $\lambda^N$ a sum, i.e.
\begin{equation*}
  \varphi^N = \varphi_1\wedge\ldots\wedge\varphi_N,\ 
  \lambda^N = \lambda_1+\cdots+\lambda_N .
\end{equation*}
Note that the definition of the wedge product contains the factor
$(N!)^{-1/2}$ whereby the
product of normalized vectors automatically becomes normalized.   
Let $H_V:=H+V$ be a second operator on $\hilbert$ with (simple) eigenvalues
$\mu_1<\mu_2<\cdots$ and eigenvectors $\psi_1,\psi_2,\ldots$. The operator
$H_V^N$ is defined analogously to $H^N$ and thus the new ground state and its
energy are
\begin{equation*}
  \psi^N = \psi_1\wedge\ldots\wedge\psi_N,\
  \mu^N = \mu_1+\cdots+\mu_N.
\end{equation*}
The transition probability, $\mathcal{D}_N$, studied by Anderson is given through
the scalar product
\begin{equation}\label{transition_probability}
   \mathcal{D}_N := |(\varphi^N,\psi^N)|^2 
                  = |\det( ( \varphi_j,\psi_k) )_{j,k=1,\ldots, N}|^2 .
\end{equation} 
It can be estimated (see \ref{det03}) as
\begin{equation}\label{anderson_inequality}
  \mathcal{D}_N  \leq e^{-\mathcal{I}_N},\ 
   \mathcal{I}_N :=\sum_{j=1}^N\sum_{k=N+1}^\infty |(\varphi_j,\psi_k)|^2 .
\end{equation}
Here, $\mathcal{I}_N$ is the so-called 'Anderson integral' which is the object of
our main interest. The asymptotics we wish to analyze involves a second
parameter $L$ reflecting the system length so that
$\hilbert=\hilbert_L=L^2([0,L]^d)$ is the Hilbert space of (spin-less) fermions
confined to the box $[0,L]^d$.
Therefore, we work with a sequence of Hilbert spaces $\hilbert_L$ and ground
states $\varphi^N=\varphi^N_L$, $\psi^N=\psi^N_L$ with $L>0$. In
the thermodynamic limit we let $N,L\to\infty$ with the particle density $\rho=N/L^d$
being kept fixed. The main result (Theorem \ref{asymptotics_d02t}) is an asymptotic
formula for the Anderson integral
\begin{equation*}
  \mathcal{I}_{N,L} = \gamma \ln N + O(1),\ N,L\to\infty,
\end{equation*}
in dimension $d=1$ and with a slightly different convention for the box size
(namely $2L$ instead of $L$) and the density $\rho=(N+\frac{1}{2})/(2L)$. The
coefficient can be computed explicitly, Corollary \ref{asymptotics_d03t},
\begin{equation*}
   \gamma(\nu) = \frac{1}{\pi^2}(1- \re t(\sqrt{\nu})),\ \nu:=\pi^2\rho^2,
\end{equation*}
where $t(\sqrt{\nu})$ is the transmission coefficient at energy $\nu$ (cf.
\cite{DeiftTrubowitz1979}). Scattering theory tells us (see
\cite{DeiftTrubowitz1979}, \cite{Lubenets1989}) that usually $\gamma(\nu)>0$ in which case
the transition probability behaves precisely as (Corollary \ref{det02t}) 
\begin{equation*}
  \tilde CN^{-\tilde\gamma(\nu)}\leq\mathcal{D}_{N,L}\leq CN^{-\gamma(\nu)},\
  N,L\to\infty .
\end{equation*} 
Here, $\tilde\gamma(\nu)>0$ can be derived from $\gamma(\nu)$.

The main ingredient of the proof is an integral formula for
$\mathcal{I}_{N,L}$ (Proposition \ref{anderson01t}), which holds true under
rather general conditions. It rests essentially upon the Riesz integral formula
for spectral projections and Krein's resolvent formula. In order to adapt it to
Schr\"odinger operators we derive a resolvent formula involving abstract
differentiation and multiplication operators (Proposition \ref{anderson02t}).
Via this formula, a sequence of scalar functions comes into play which tends at
least informally to a Dirac delta function. This is made precise in Sections
\ref{f_delta} and \ref{delta}, hence the name delta-term and delta-estimate. The
singularity represented by the delta sequence
reflects in a way the singular transition from a discrete spectrum to a
continuous spectrum as $L\to\infty$.

Our method requires a rather detailed and precise knowledge of the free
Dirichlet problem, in particular of the resolvent. Almost everything one needs
to know about the perturbed problem, however, can be read off from the so-called
T-operator. The perturbed eigenvalues do not enter in the actual asymptotic
analysis. We only need to make sure that the number of perturbed eigenvalues
below some fixed (Fermi) energy is asymptotically the same for large $N$ as for
the free problem (see Proposition \ref{p_eigenvalues01t}). This is related to
the spectral shift function (see \cite{HislopMuller2010} for potentials with
compact support). Interestingly, a lot of work has been done to derive
asymptotic formulae for the perturbed eigenvalues at large energies. Except for
\cite{AkulenkoNesterov2006}, we are not aware of studies that include also the
dependence on $L$ as well.

Anderson's orthogonality catastrophe has attracted a lot of interest in solid
state physics since its discovery. There are early attempts to determine the
exact asymptotics of the determinant $\mathcal{D}_N$ itself. 
Rivier and Simanek \cite{RivierSimanek1971} used the
adiabatic theorem to express $\mathcal{D}_N$ through the solution of a
Wiener-Hopf equation. However, they could not deal satisfactorily with certain
limit procedures underlying the method.
This was improved upon by Hamann \cite{Hamann1971} who, likewise, could treat the thermodynamic
limit only informally. A clarification of that method can be found in \cite{Otte2005}.
Recent numerical investigations have been carried out by
Weichselbaum, M\"under, and von Delft \cite{WeichselbaumMunderDelft2011} who
also present some physical background and refer to further reading.  

Frank, Lewin, Lieb, and Seiringer \cite[Eq. (11)]{FrankLewinLiebSeiringer2011}
considered the related problem of proving a lower bound to the energy difference
-- in our notation below $\tr(H_V\Pi-HP)$ -- directly in the thermodynamic limit
in terms of semi-classical quantities.

Gebert, K\"uttler, and M\"uller \cite{GebertKuettlerMueller2013} using different
methods have recently established a rigorous lower bound
\begin{equation*}
  \mathcal{I}_{N,L} \geq \gamma' \ln N,\ N,L\to\infty,
\end{equation*}
in any dimension (even with a periodic background potential but with positive
and compactly supported exterior potential $V$). Remarkably, their value
$\gamma'$ agrees with Anderson's prediction. In our framework, the expression
for $\gamma$ that at first came out from Theorem \ref{asymptotics_d02t} is rather
implicit. Only after some computation could we confirm that $\gamma=\gamma'$,
Corollary \ref{asymptotics_d03t}. Thus, one can reasonably conjecture that
$\gamma'\ln N$ is indeed the exact leading asymptotics in any dimension. 

%% file: anderson.tex
Let $\hilbert$ be a Hilbert space with scalar product $(\cdot,\cdot)$, which is
anti-linear in the first component and linear in the second component and let $\|\cdot\|$ be
the corresponding vector norm. The induced operator norm will be denoted with the same 
symbol $\|\cdot\|$. 
We consider a self-adjoint operator $H:D(H)\to\hilbert$, $D(H)\subset\hilbert$,
and a bounded operator $V:\hilbert\to\hilbert$. Then,
$H_V=H+V$ is self-adjoint as well with $D(H_V)=D(H)$. We denote by $\sigma(H)$ and 
$\sigma(H_V)$ the spectrum of $H$ and $H_V$, respectively, and by
\begin{equation}\label{anderson_resolvents}
  R(z):= (z\id-H)^{-1},\ z\in\C\setminus\sigma(H),\
  R_V(z) :=(z\id-H_V)^{-1},\ z\in\C\setminus\sigma(H_V)
\end{equation}
their resolvents. From spectral theory we know
\begin{equation}\label{anderson_resolvents_norm}
  \| R(z) \| = \frac{1}{\mathrm{dist}(z,\sigma(H))},\
  \| R_V(z) \| = \frac{1}{\mathrm{dist}(z,\sigma(H_V))} .
\end{equation}
We borrow some notation from scattering theory (see
e.g. \cite[3.6]{Thirring2002}). Note, that for $z\notin\sigma(H)$ the operator
$(\id-VR(z))^{-1}$ exists and is bounded if and only if
$z\notin\sigma(H_V)$. The same holds true for $(\id-R(z)V)^{-1}$. Hence, the
so-called transition operator or T-operator
\begin{equation}\label{anderson_T-matrix}
  T(z) := (\id-VR(z))^{-1}V =  V(\id-R(z)V)^{-1}
\end{equation}
exists for $z\in\C\setminus(\sigma(H)\cup\sigma(H_V))$ with
\begin{equation}\label{anderson_T-matrix_norm}
  \| T(z) \| \leq \frac{\mathrm{dist}(z,\sigma(H))}{\mathrm{dist}(z,\sigma(H_V))} \|V\|
\end{equation}
and is analytic there as a function of $z$. Krein's resolvent formula
\begin{equation}\label{anderson_krein}
  R_V(z) - R(z) = R(z) T(z) R(z)
\end{equation}
relates the resolvents $R(z)$ and $R_V(z)$ with each other whenever the T-operator
exists. The operator $V$ plays an important role via its modified polar decomposition
\begin{equation}\label{anderson_V_polar}
  V = \sqrt{|V|}J\sqrt{|V|},\ J^*=J,\ J^2=\id,\ \|J\|=1,
\end{equation}
which is obvious for the multiplication operators used below.
Like in scattering theory it is advantageous to look at operators relative to
$V$. More precisely, we
will use (cf. \cite[3.6.1, 1]{Thirring2002})
\begin{equation}\label{anderson_Omega}
   \sqrt{|V|}R(z)\sqrt{|V|},\
      \Omega(z) := (\id-\sqrt{|V|}R(z)\sqrt{|V|}J)^{-1}
\end{equation}
with the sandwiched resolvent being called Birman-Schwinger operator. Note the
relation
\begin{equation}\label{anderson_Omega_T}
   T(z) = \sqrt{|V|} J\Omega(z)\sqrt{|V|} .
\end{equation}
Obviously, the Birman-Schwinger operator exists and is bounded for
$z\in\C\setminus\sigma(H)$. For $\Omega(z)$ to exist as a bounded operator it is
required that $z\notin\C\setminus(\sigma(H)\cup\sigma(H_V))$. The converse is
true, too. That is to say, if $z\notin\sigma(H)$ and $\Omega(z)$ exists and is
bounded then $z\notin\sigma(H_V)$. In order to see this one first shows that
$\id-R(z)V$ is injective and has dense range and, in a second step, that the
range is closed.

%% file: anderson_spectral_gap.tex
Riesz's integral formula yields a handy expression for the Anderson integral when
the operators $H$ and $H_V$ have a common spectral gap. That is to say their
spectra can be written as
\begin{equation}\label{anderson_spectrum}
  \sigma(H) = \sigma_1(H) \cup \sigma_2(H),\
  \sigma(H_V) = \sigma_1(H_V) \cup \sigma_2(H_V)
\end{equation}
such that there is a closed contour $\Gamma\subset\C$ with each $\sigma_1$ being
inside and each $\sigma_2$ outside of $\Gamma$. Let $P$ be the spectral
projection of $H$ belonging to $\sigma_1(H)$ and let $\Pi$ be defined likewise for
$H_V$. The Anderson integral in question is
\begin{equation}\label{anderson_integral}
  \mathcal{I} := \tr \big[P(\id-\Pi)\big] .
\end{equation}
In our application $P$ is trace class and hence $0\leq\mathcal{I}<\infty$. The Riesz
formula reads
\begin{equation}\label{anderson_riesz}
  P = \frac{1}{2\pi i}\int_\Gamma R(z)\, dz,\
  \Pi = \frac{1}{2\pi i} \int_\Gamma R_V(z)\, dz .
\end{equation}
Note that both integrals have the same $\Gamma$ from above. For our purposes,
an infinite contour is more appropriate. In particular, due to the special form
of the free Green function (see \eqref{f_resolvent_green}) a parabola will do
best.

\begin{proposition}\label{anderson01t}
Let $P$ be trace class. We assume the sets $\sigma_{1,2}$ in \eqref{anderson_spectrum} to satisfy
\begin{equation}\label{anderson01t01}
   \sup\sigma_1(H) < \nu < \inf\sigma_2(H),\
   \sup\sigma_1(H_V) < \nu < \inf\sigma_2(H_V)
\end{equation}
with some $\nu\in\R$ and define the parabola $\Gamma_\nu:= \{z=(\sqrt{\nu}+is)^2\mid s\in\R\}$.
Then, the difference of the spectral projections has the representation
\begin{equation}\label{anderson01t02a}
  \Pi - P = \frac{1}{2\pi i} \int_{\Gamma_\nu} R(z) T(z) R(z)\, dz
\end{equation}
and the Anderson integral \eqref{anderson_integral} can be written as
\begin{equation}\label{anderson01t02}
  \mathcal{I}
    =  \frac{1}{2\pi i} \int_{\Gamma_\nu} \tr \big[ PR(z) T(z) R(z)^2 T(z) \big]
    \, dz .
\end{equation}
\end{proposition}
\begin{proof}
By Riesz's and Krein's formulae,
\eqref{anderson_riesz} and \eqref{anderson_krein},
\begin{equation}\label{anderson01t03}
  \Pi - P 
      = \frac{1}{2\pi i} \int_\Gamma ( R_V(z) - R(z) ) \, dz
      = \frac{1}{2\pi i} \int_\Gamma R(z) T(z) R(z)\, dz
\end{equation}
with the closed contour $\Gamma$ used in \eqref{anderson_spectrum}.
For the Anderson integral note $P(\Pi-\id) =  P(\Pi-P)$ which allows us to use
\eqref{anderson01t03}. Since $P$ is trace class and the other operators are
bounded we may take the trace. Using the cyclic commutativity we obtain
\begin{equation*}
  \mathcal{I}
    = -\frac{1}{2\pi i}\int_\Gamma
       \tr\big[ P R(z) T(z) R(z)\big] \, dz
    = -\frac{1}{2\pi i}\int_\Gamma
       \tr \big[P R(z)^2 T(z)\big] \, dz
\end{equation*}
since $P$ commutes with $R(z)$. Recall that $R(z)$ is differentiable in
$z\in\C\setminus\sigma(H)$ with $R'(z)=-R(z)^2$. Since all functions involved
are analytic for $z\in\C\setminus(\sigma(H)\cup\sigma(H_V))$ we may integrate by
parts,
\begin{equation}\label{anderson01t04}
    \mathcal{I}
      = -\frac{1}{2\pi i}\int_\Gamma \tr\big[ P R(z) T'(z)\big]\, dz
      = \frac{1}{2\pi i}\int_\Gamma \tr\big[ P R(z) T(z)R(z)^2T(z)
        \big]\, dz .
\end{equation}
By the estimates \eqref{anderson_resolvents_norm} and \eqref{anderson_T-matrix_norm} the
integrands in \eqref{anderson01t03} and \eqref{anderson01t04} decay fast enough
at infinity so that we may bend the closed contour $\Gamma$ into the parabola
$\Gamma_\nu$ to obtain \eqref{anderson01t02a} and \eqref{anderson01t02}, respectively.
\end{proof}

The integral formula \eqref{anderson01t02} for the Anderson integral was
intentionally made more complicated via integration by parts. For, in the
application of the delta-estimate to \eqref{anderson01t02} it will be important
to have the smooth cut-off factor $PR(z)$ instead of just $P$.

%% file: anderson_schroedinger_type.tex
A typical Schr\"odinger operator is built from differentiation and multiplication
operators. Let us introduce two operators $\nabla$ and $X$ satisfying
\begin{equation}\label{anderson_ccr}
  [\nabla,X] = \id.
\end{equation}
We assume $\nabla:D(\nabla)\to\hilbert$, $D(\nabla)\subset\hilbert$, 
to be densely defined on $\hilbert$ and
$X:\hilbert\to\hilbert$ to be bounded such that $XD(\nabla)\subset
D(\nabla)$. Thus, \eqref{anderson_ccr} is meant to hold true on $D(\nabla)$.
Self-adjointness of Schr\"odinger operators often results from boundary
conditions which usually lessen the domain of definition. Let
$-\nabla^2$ have a self-adjoint restriction $H:D(H)\to\hilbert$, i.e.
$D(H)\subset D(\nabla^2)$ and
\begin{equation}\label{anderson_H}
  H = -\nabla^2\ \text{on}\ D(H) .
\end{equation}
The resolvent of $H$,
\begin{equation*}
  R(z) = (z\id-H)^{-1},\ z\in\C\setminus\sigma(H)
\end{equation*}
is a well-defined and bounded operator with $R(z):\hilbert\to D(H)$. The latter
implies
\begin{equation}\label{anderson_inverse}
  (z\id+\nabla^2) R(z) = (z\id-H)R(z) = \id.
\end{equation}
In general, this equality fails to hold true when the order of terms is
switched as can be seen in Proposition \ref{f_delta01t}. This is the reason why the
following resolvent formula gives non-trivial results.

\begin{proposition}\label{anderson02t}
For operators $\nabla$ and $X$ as in \eqref{anderson_ccr} let us assume in
addition $XD(H)\subset D(H)$. Then, the decomposition
\begin{equation}\label{anderson02t01}
  R(z)^2  = \frac{1}{z}R(z) - \frac{1}{2z} [X\nabla,R(z)] + \frac{1}{z} D(z)
          = \frac{1}{2z} ( R(z) - C(z) ) + \frac{1}{z} D(z)
\end{equation}
holds true on $D(\nabla^2)$ and for $z\in\C\setminus\sigma(H)$. Here,
\begin{equation}\label{anderson02t02}
  D(z) := \Big(\frac{1}{2} X - R(z)\nabla\Big)[\nabla,R(z)],\
  C(z) := X\nabla R(z) - R(z)\nabla X .
\end{equation}
The operator $D(z)$ is the so-called 'delta-term' and satisfies
\begin{equation}\label{anderson02t03}
  (z\id + \nabla^2) D(z) = 0
\end{equation}
on $D(\nabla^2)$ and for $z\in\C\setminus\sigma(H)$.
\end{proposition}
\begin{proof}
We start off from the elementary formula
\begin{equation}\label{anderson02t03a}
  R(z)^2  = \frac{1}{z} R(z) + \frac{1}{z}R(z)HR(z)
\end{equation}
and rewrite the last term. By the product rule the commutator in
\eqref{anderson02t01}
becomes
\begin{equation}\label{anderson02t04}
  [X\nabla,R(z)] = X[\nabla,R(z)] + [X,R(z)]\nabla .
\end{equation}
Formula \eqref{anderson_ccr} implies $[\nabla^2, X] = 2 \nabla$. By noting $XD(H)\subset
D(H)$ and $R(z)(z\id-H)=\id$ on $D(H)$ we obtain
\begin{equation*}
  [X,R(z)]  = R(z)[z\id-H,X]R(z)
            = R(z)[\nabla^2,X]R(z)
            = 2R(z)\nabla R(z) .
\end{equation*}
Thus,
\begin{equation*}
  \frac{1}{2}[X,R(z)]\nabla
       = R(z)\nabla R(z)\nabla
       = R(z)\nabla^2 R(z) + R(z)\nabla[R(z),\nabla]
\end{equation*}
Recalling \eqref{anderson_H} we solve for $R(z)HR(z)$, insert this into
\eqref{anderson02t03a}, and use \eqref{anderson02t04}. Then,
\begin{equation*}
\begin{split}
  R(z)^2 & = \frac{1}{z} R(z) - \frac{1}{2z} [X,R(z)]\nabla - \frac{1}{z}R(z)\nabla[\nabla,R(z)] \\
         & = \frac{1}{z} R(z) - \frac{1}{2z}[X\nabla,R(z)] + \frac{1}{2z}
         X[\nabla,R(z)] - \frac{1}{z}R(z)\nabla[\nabla,R(z)] .
\end{split}
\end{equation*}
With the definition \eqref{anderson02t02} of $D(z)$ this is the first equality in \eqref{anderson02t01}. 
The second one follows by means of the commutation relation $X\nabla=-\id+\nabla X$.
Finally, by \eqref{anderson_ccr}
\begin{equation*}
  (z\id +\nabla^2) ( \frac{1}{2}X - R(z)\nabla)
     = \frac{1}{2} X(z\id +\nabla^2) .
\end{equation*}
Then,
\begin{equation*}
  (z\id+\nabla^2)[\nabla,R(z)] 
   = \nabla (z\id+\nabla^2) R(z) - (z\id+\nabla^2) R(z)\nabla
   = 0
\end{equation*}
shows \eqref{anderson02t03}.
\end{proof}

Our motivation behind the resolvent formula \eqref{anderson02t01} in Proposition
\ref{anderson02t} is that it splits the integrand
$\tr[PR(z)T(z)R(z)^2T(z)]$ in the integral representation of the
Anderson integral, Proposition \ref{anderson01t}, into a sum of two terms. The
first term, $\tr[PR(z)T(z)(R(z)-C(z))T(z)]$, will be subdominant, i.e.
$O(1)$, as shown in Section \ref{asymptotics_s} whereas the second term
$\tr[PR(z)T(z)D(z)T(z)]$ is of the leading order $\ln N$, see Section
\ref{asymptotics_d}. The operator $D(z)$ quantifies the difference between the
resolvent of the Laplace operator with and without Dirichlet boundary
conditions.

%% file: schroedinger_operators.tex
We look into the special case of Schr\"odinger operators with Dirichlet boundary
conditions on the finite interval $[-L,L]$. Our Hilbert space then becomes
$\hilbert=L^2[-L,L]$. Actually, it ought to bear an index $L$ as well as all
operators defined on it and related quantities. However, since this dependence
is ubiquitous we tacitly suppress it. In our concrete case,
\begin{equation}\label{nabla_X}
  \nabla = \frac{d}{dx},\
  (X\varphi)(x) = x\varphi(x).
\end{equation}
The domain $D(\nabla)$ as well as $D(\nabla^2)$ can be described with the aid of
Sobolev spaces which we do not need in detail herein. One can show that
$X(D(\nabla))\subset D(\nabla)$. The operator $H$ becomes
\begin{equation*}
  H = -\nabla^2 = -\frac{d^2}{dx^2}\ \text{on}\ D(H),
\end{equation*}
where $D(H)$ is $D(\nabla^2)$ restricted by Dirichlet boundary conditions.
Because of that we have $XD(H)\subset D(H)$. The corresponding
eigenvalue problem reads
\begin{equation}\label{free_eigenvalue_problem}
  -\varphi'' = \lambda \varphi,\ \varphi(-L) = 0 = \varphi(L) .
\end{equation}
The eigenvalues $\lambda_j$ and normalized eigenfunctions $\varphi_j$, $j\in\N$, are
\begin{equation}\label{free_spectral_quantities}
\begin{aligned}
  \lambda_j = \big(\frac{\pi j}{2L}\big)^2,
    & &
  \varphi_j(x) =
\begin{cases}
  \frac{1}{\sqrt{L}} \sin(\frac{\pi j}{2L}x) & \text{for}\ j\ \text{even}, \\
  \frac{1}{\sqrt{L}} \cos(\frac{\pi j}{2L}x) & \text{for}\ j\ \text{odd} .
\end{cases}
\end{aligned}
\end{equation}
We translate the integral formula in Proposition \ref{anderson01t} and
the resolvent formula \eqref{anderson02t01} into the framework of Schr\"odinger
operators. For the $\nu\in\R$ in Proposition \ref{anderson01t}
separating the two parts of the spectrum we choose the so-called Fermi energy
\begin{equation}\label{fermi_energy}
  \nu_N:=\Big[\frac{\pi}{2L}(N+\frac{1}{2})\Big]^2 .
\end{equation}
Thereby, the spectrum of $H$ decomposes into $\sigma(H)=\sigma_1(H)\cup\sigma_2(H)$,
\begin{equation*}
  \sigma_1(H) := \{\lambda_j \mid 1\leq j \leq N\},\
  \sigma_2(H) := \{\lambda_j \mid j\geq N+1\}
\end{equation*}
and the parabola $\Gamma_{\nu_N}$ becomes what we call Fermi parabola
\begin{equation}\label{fermi_parabola}
  \Gamma_N := \{ z = (\sqrt{\nu_N}+is)^2 \mid -\infty<s<\infty\},\ 
  dz = 2i(\sqrt{\nu_N}+is)\, ds .
\end{equation}
The distance of the Fermi parabola from the spectrum is
\begin{equation}\label{fermi_parabola_distance}
  |z-\lambda_j|
    = |\sqrt{\nu_N}+is+\sqrt{\lambda_j}||\sqrt{\nu_N}+is-\sqrt{\lambda_j}|
    \geq (\nu_N + s^2)^{\frac{1}{2}}
                            ((\sqrt{\nu_N}-\sqrt{\lambda_j})^2+s^2)^{\frac{1}{2}} ,
\end{equation}
which will be used at various points in particular with $s=0$. The spectral
projection $P$ in the Anderson integral \eqref{anderson_integral} becomes
\begin{equation}\label{free_spectral_projection}
  P_N := \sum_{j=1}^N (\varphi_j,\cdot)\varphi_j .
\end{equation}
The perturbed operator $H_V$ is given by
\begin{equation*}
  H_V = H+V ,
\end{equation*}
where $V$ is the operator of multiplication by a real-valued function $V$, the potential,
denoted by the same symbol for the sake of simplicity. Some results further
below will be uniform in $L$. In order to formulate this conveniently we assume that the
potential $V$ is already defined on the whole of $\R$ and not only on the
interval $[-L,L]$. Thus, we denote by $\|V\|_r$, $1\leq r\leq \infty$ the $L^r(\R)$
norms of the function $V$.
If $V\in L^\infty(\R)$ then the operator $V$ is bounded
regardless of $L$, which is in line with Section \ref{anderson}. In particular,
$D(H_V)=D(H)$. Furthermore, since the free eigenfunctions are obviously
delocalized, $V\in L^1(\R)$ implies
\begin{equation}\label{free_eigenfunction_delocalized} 
  \|\sqrt{|V|}\varphi_j\| \leq \frac{1}{\sqrt{L}} \|V\|_1^{\frac{1}{2}} ,
\end{equation}
which will be used throughout.
The spectrum of $H_V$ is given through the corresponding Dirichlet problem
\begin{equation}\label{perturbed_eigenvalue_problem}
  -\psi'' + V\psi = \mu \psi,\ \psi(-L) = 0 = \psi(L) .
\end{equation}
It consists solely of simple eigenvalues, which follows easily
via uniqueness results for ordinary differential equations. We denote them by
$\mu_k$, $k\in\N$ with the usual ordering $\mu_1<\mu_2<\cdots$. 
The decomposition \eqref{anderson_spectrum} of $\sigma(H_V)$ will be studied in Section
\ref{p_resolvent}. The normalized eigenfunctions of $H_V$ are $\psi_k$, $k\in\N$
and the spectral projection $\Pi$ in \eqref{anderson_integral} reads
\begin{equation}\label{perturbed_spectral_projection}
  \Pi_M := \sum_{k=1}^M (\psi_k,\cdot)\psi_k .
\end{equation}
Note that in general $M\neq N$ (see Section \ref{p_eigenvalues}).

%% file: free_resolvent.tex
The spectral representation of the free resolvent \eqref{anderson_resolvents}
with \eqref{free_spectral_quantities} reads
\begin{equation}\label{f_resolvent_resolvent}
  R(z) = \sum_{j=1}^\infty \frac{1}{z-\lambda_j} (\varphi_j,\cdot)\varphi_j .
\end{equation}
The corresponding kernel or Green function is given by
\begin{equation}\label{f_resolvent_green}
  R(z;x,y) =   \frac{1}{W(z)}
\begin{cases}
 \sin(\sqrt{z}(x-L))\sin(\sqrt{z}(y+L)) & -L\leq y\leq x \leq L\\
 \sin(\sqrt{z}(x+L))\sin(\sqrt{z}(y-L)) & -L\leq x\leq y \leq L
\end{cases}
\end{equation}
with the Wronski determinant
\begin{equation}\label{f_resolvent_wronski}
  W(z) = 2\sqrt{z}\sin( L\sqrt{z})\cos( L\sqrt{z}) = \sqrt{z}\sin(2L\sqrt{z}) .
\end{equation}
By rewriting the Green function one can cast the resolvent into a form where the
$L$ dependence is more tangible
\begin{equation}\label{f_resolvent_L}
  R(z) = \frac{1}{2\sqrt{z}}
    \Big[ \frac{\cos(L\sqrt{z})}{\sin(L\sqrt{z})} P_s(z) 
          - \frac{\sin(L\sqrt{z})}{\cos(L\sqrt{z})} P_c(z)
          + G(z) \Big] .
\end{equation}
The operators $P_s(z)$, $P_c(z)$, and $G(z)$ have the kernels
\begin{equation}\label{f_resolvent_P}
\begin{gathered}
    P_s(z;x,y) := \sin(\sqrt{z}x)\sin(\sqrt{z}y),\
    P_c(z;x,y) := \cos(\sqrt{z}x)\cos(\sqrt{z}y),\\
    G(z;x,y) := \sin(\sqrt{z}|x-y|) .
\end{gathered}
\end{equation}
Note that $P_s(z)$ and $P_c(z)$ are rank-one operators which makes the resolvent differ from 
the operator $G(z)$ by a rank-two perturbation. 
We would like to apply the delta-estimate from Section \ref{delta} directly to
$R(z)$ and $\Omega(z)$ (cf. \eqref{anderson_Omega}). However, the prefactors
of $P_s(z)$ and $P_c(z)$ in \eqref{f_resolvent_L} behave too singularly at
$z=\nu_N$ to do that. In a first step we therefore replace $z$ in the
benevolent operators $P_{s,c}(z)$ and $G(z)$ by $\nu_N$ and retain the
malevolent dependence in the function $\tau$. This motivates the definition of
the operators $R_\infty^\pm(\nu_N,Ls)$ and $\Omega_\infty^\pm(\nu_N,Ls)$ in
\eqref{f_resolvent_limit} and \eqref{p_resolvent_limit}, respectively. In
\eqref{f_resolvent02t06}, we estimate the difference between $R(z)$ and
$R_\infty^\pm(\nu_N,Ls)$. Later, in our main Theorem \ref{asymptotics_d02t}, we
use these operators to compute the coefficient of the leading asymptotic
$N$-behaviour of the Anderson integral. To begin with, we have a closer look
at \eqref{f_resolvent_L}. At the Fermi energy \eqref{fermi_energy}
\begin{equation}\label{f_resolvent_sine_cosine}
  \sin (L\sqrt{\nu_N}) = \frac{1}{\sqrt{2}} (-1)^{\lfloor\frac{N}{2}\rfloor},\
  \cos (L\sqrt{\nu_N}) = \frac{1}{\sqrt{2}} (-1)^{\lceil\frac{N}{2}\rceil} ,
\end{equation}
which implies on the Fermi parabola \eqref{fermi_parabola}
\begin{equation}\label{sine_cosine_fermi_parabola}
\begin{aligned}
  \sin(L(\sqrt{\nu_N}+is)) 
     & = \frac{1}{\sqrt{2}}(-1)^{\lfloor\frac{N}{2}\rfloor}(\cosh(Ls) +
         i(-1)^N\sinh(Ls)) , \\
  \cos(L(\sqrt{\nu_N}+is))
     & = \frac{1}{\sqrt{2}}(-1)^{\lfloor\frac{N}{2}\rfloor}((-1)^N\cosh(Ls)
         -i\sinh(Ls)) .
\end{aligned}
\end{equation}
Furthermore, we have
\begin{equation}\label{cotangent_fermi_parabola}
  \frac{\cos(L(\sqrt{\nu_N}+is))}{\sin(L(\sqrt{\nu_N}+is))} =
   (-1)^N \tau((-1)^N Ls)
\end{equation}
where
\begin{equation}\label{f_resolvent_tau}
    \tau(s) := \frac{\cosh s-i\sinh s}{\cosh s + i\sinh s},\ \tau(-s)=\bar\tau(s),\
     |\tau(s)| =1,\ s\in\R ,\ 
    \lim_{s\to\infty} \tau(s) = - i.
\end{equation}
Now, we keep the $s$-dependence only in the scalar function $\tau$ but not in
the operators $P_{s,c}(z)$ and $G(z)$ and introduce
\begin{equation}\label{f_resolvent_limit}
  R_\infty^\pm (\nu_N,Ls) := \frac{1}{2\sqrt{\nu_N}}
     \Big[\pm\tau(\pm Ls) P_s(\nu_N) \mp \tau(\mp Ls) P_c(\nu_N) + G(\nu_N)\Big] .
\end{equation}
This can be seen, in a way, as the limit of the resolvent as $L\to\infty$ (cf. \eqref{f_resolvent02t06}). 
Note that $R_\infty^\pm(\nu_N,s)$ differs from $G(\nu_N)$ by a rank-two
perturbation.

The operator $C(z)$ in \eqref{anderson02t02} has the kernel ($x,y\in[-L,L]$)
\begin{equation}\label{f_resolvent_commutator}
   C(z;x,y)
    = \frac{\sqrt{z}}{2W(z)}
\begin{cases}
   x\cos(\sqrt{z}(x-L))\sin(\sqrt{z}(y+L)) \\
     \quad + y\sin(\sqrt{z}(x-L))\cos(\sqrt{z}(y+L))
    &  y \leq x , \\
   x\cos(\sqrt{z}(x+L))\sin(\sqrt{z}(y-L)) \\
     \quad + y\sin(\sqrt{z}(x+L))\cos(\sqrt{z}(y-L))
    &  x\leq y .
\end{cases}
\end{equation}
The Green function and related quantities are to be evaluated on the Fermi parabola
$\Gamma_N$.

\begin{lemma}\label{f_resolvent01t}
For all $s\in\R$, $L>0$, $N\in\N$, and $\nu_N$ as in \eqref{fermi_energy} we have
\begin{equation}\label{f_resolvent01t01}
  \frac{1}{|\sin (L(\sqrt{\nu_N}+is))|^2} \leq 4 e^{-2L|s|},\
  \frac{1}{|\cos (L(\sqrt{\nu_N}+is))|^2} \leq 4 e^{-2L|s|}.
\end{equation}
Moreover, for $z\in\Gamma_N$ (see \eqref{fermi_parabola})
\begin{equation}\label{f_resolvent01t02}
  |R(z;x,y)| \leq \frac{2}{(\nu_N+s^2)^{\frac{1}{2}}} e^{-|s| |x-y|} ,\
  |C(z;x,y)| \leq 2 (|x|+|y|) e^{-|s||x-y|}.
\end{equation}
Let $z=(a+is)^2$, $a,s\in\R$. Then, the kernels of the operators $P_{s,c}(z)$
and $G(z)$ from \eqref{f_resolvent_P} satisfy
\begin{align*}
  |P_{s,c}(z;x,y) - P_{s,c}(a^2;x,y)| & \leq |s| ( |x|+|y| ) e^{|s|( |x|+|y| )},\\
  | G(z;x,y) - G(a^2;x,y) | & \leq |s| |x-y| e^{|s| |x-y|} .
\end{align*}
\end{lemma}
\begin{proof}
(a)
From \eqref{sine_cosine_fermi_parabola} we deduce
\begin{equation*}
  |\sin(L(\sqrt{\nu_N}+is))|^2
    = |\cos(L(\sqrt{\nu_N}+is))|^2
    = \frac{1}{4}(e^{2Ls} + e^{-2Ls} )
    \geq \frac{1}{4} e^{2L|s|}
\end{equation*}
which proves \eqref{f_resolvent01t01}.

(b) 
For $L\geq x\geq y \geq -L$
\begin{equation*}
  |\sin((\sqrt{\nu_N}+is)(x-L))\sin((\sqrt{\nu_N}+is)(y+L))|
    \leq e^{|s|(2L - |x-y|)} ,
\end{equation*}
where we estimated the sine by the exponential function.
For $x\leq y$ the bound looks the same. Using \eqref{f_resolvent01t01} we obtain
\begin{equation*}
\begin{split}
  | R(z;x,y)|
    & \leq \frac{1}{(\nu_N+s^2)^{\frac{1}{2}}} 
          \frac{e^{|s|(2L-|x-y|)}}{2|\sin (L(\sqrt{\nu_N}+is))| |\cos (L(\sqrt{\nu_N}+is))|}
            \\
    & \leq \frac{2}{(\nu_N+s^2)^{\frac{1}{2}}} e^{-|s||x-y|} ,
\end{split}
\end{equation*}
which proves the first estimate in \eqref{f_resolvent01t02}. The estimate for
$C(z;x,y)$ in \eqref{f_resolvent01t02} follows likewise.

(c)
We write the difference as an integral
\begin{equation*}
\begin{split}
  \lefteqn{P_s(z;x,y) - P_s(a^2;x,y) }\\
   & = \int_0^s \frac{d}{dt}\big(
        \sin((a+it)x)\sin((a+it)y)\big)\, dt\\
   & = i\int_0^s \big( x\cos((a+it)x)\sin((a+it)y)
                    +y\sin((a+it)x)\cos((a+it)y)\big)\, dt,
\end{split}
\end{equation*}
and estimate
\begin{equation*}
  |P_s(z;x,y) - P_s(a^2;x,y)|
     \leq \int_0^{|s|} (|x|+|y|)e^{t(|x|+|y|)}\, dt
     \leq |s| ( |x|+|y| ) e^{|s|(|x|+|y|)} .
\end{equation*}
The estimates for $P_c(z)$ and $G(z)$ follow in like manner.
\end{proof}

Similar to the Birman-Schwinger operator \eqref{anderson_Omega} we
need to study operators of the form $\sqrt{|V|}P_{s,c}(z)\sqrt{|V|}$.
To this end, we introduce the functions $\omega_s(z)$, $\omega_c(z)$,
\begin{equation}\label{f_resolvent_omega}
   \omega_s(z;x) := \sqrt{|V(x)|}\sin(\sqrt{z}x) ,\ 
   \omega_c(z,x) := \sqrt{|V(x)|}\cos(\sqrt{z}x) ,
\end{equation}
$z\in\C$, $x\in\R,$ so that the kernels read (cf. \eqref{f_resolvent_P})
\begin{equation*}
  \sqrt{|V|}P_{s,c}(z)\sqrt{|V|}(x,y) = \omega_{s,c}(z;x)\omega_{s,c}(z;y) .
\end{equation*}
In order to describe how $\omega_{s,c}(z)$ and derived quantities
behave in the complex plane we associate to any $V\in L^1(\R)$ the transformed function 
$V_L\in C^\infty(\R)$,
\begin{equation}\label{V_L}
   V_L(s) := \int_{-L}^L |V(x)| e^{s|x|}\, dx ,\ s\in\R .
\end{equation} 
Its derivatives satisfy
\begin{equation}\label{V_L_estimates}
\begin{gathered}
    0 \leq \ V_L^{(p)}(0) \leq V_L^{(p)}(s),\
    V_L^{(p+q)}(0) \leq L^p \|X^qV\|_1,\\
    V_L^{(p+q)}(s) \leq L^p \|X^qV\|_\infty \int_{-L}^Le^{s|x|}\, dx
\end{gathered}
\end{equation}
with $p,q\in\N_0$ provided that $X^qV\in L^1(\R)$ and $X^qV\in L^\infty(\R)$,
respectively.

\begin{lemma}\label{f_resolvent01bt}
Let $V\in L^1(\R)$ and $z=(a+is)^2$ with $a,s\in\R$. Then,
\begin{align}
  \|\omega_{s,c}(z)\|
     & \leq V_L(2|s|)^{\frac{1}{2}},
            \label{f_resolvent01bt02} \\
  \|\omega_{s,c}(z) - \omega_{s,c}(a^2)\|
     & \leq |s| V_L^{(2)}(2|s|)^{\frac{1}{2}} .
           \label{f_resolvent01bt03}
\end{align}
\end{lemma}
\begin{proof}
In order to prove \eqref{f_resolvent01bt02} we estimate
\begin{equation*} 
  \|\omega_s(z)\|^2 
     = \int_{-L}^L |V(x)| |\sin((a+is)x)|^2\, dx
     \leq \int_{-L}^L |V(x)| e^{2|sx|} \, dx .
\end{equation*}
For \eqref{f_resolvent01bt03} we compute
\begin{equation*}
   \|\omega_s(z)-\omega_s(a^2)\|^2
      = \int_{-L}^L |V(x)|\cdot |\sin((a+is)x)-\sin(ax)|^2\, dx
\end{equation*}
and use the estimate
\begin{equation*}
  |\sin((a+is)x) - \sin(ax)|
    = |ix \int_0^s \cos((a+it)x)\, dt|
    \leq |x||s| e^{|sx|} ,
\end{equation*}
which yields \eqref{f_resolvent01bt03}. The estimates for $\omega_c(z)$ follow
in like manner.
\end{proof}

One could use Lemma \ref{f_resolvent01t} to study the norms of
$R(z)$ or $G(z)$. However, the applications we have in mind require that to be
done for the Birman-Schwinger (see \eqref{anderson_Omega}) and
suchlike operators (with $\sqrt{|V|}$ multiplied from left and right).

\begin{lemma}\label{f_resolvent02t}
Let $V\in L^1(\R)$ and $z\in\Gamma_N$. Then, the Birman-Schwinger operator satisfies
\begin{equation}\label{f_resolvent02t01}
  \| \sqrt{|V|}R(z)\sqrt{|V|} \| \leq \frac{4}{\sqrt{\nu_N+s^2}} \|V\|_1 .
\end{equation}
If $X^2V\in L^1(\R)$ with $X$ as in \eqref{nabla_X} the operator $C(z)$ from
\eqref{f_resolvent_commutator} satisfies
\begin{equation}\label{f_resolvent02t03}
  \|\sqrt{|V|} C(z) \sqrt{|V|} \| \leq 8
  \|X^2V\|_1^{\frac{1}{2}}\|V\|_1^{\frac{1}{2}} .
\end{equation}
Furthermore, for the operators $G(\nu_N)$ from \eqref{f_resolvent_P} and
$R_\infty^\pm(\nu_N,s)$ from \eqref{f_resolvent_limit} we have
\begin{equation}\label{f_resolvent02t04}
  \|\sqrt{|V|}G(\nu_N)\sqrt{|V|}\| \leq \|V\|_1,\
 \|\sqrt{|V|}R_\infty^\pm(\nu_N,s)\sqrt{|V|}\| \leq \frac{3}{2\sqrt{\nu_N}} \|V\|_1 .
\end{equation}
Finally,
\begin{multline}\label{f_resolvent02t06}
  \| \sqrt{|V|}(R_\infty(\nu_N,Ls) - R((\sqrt{\nu_N}+is)^2))\sqrt{|V|} \| \\
      \leq \frac{3}{2\nu_N} V_L(0) |s|
        + \frac{3}{\sqrt{\nu_N}}|s| V_L^{(2)}(2|s|)^{\frac{1}{2}} V_L(2|s|)^{\frac{1}{2}}
\end{multline}
where $R_\infty$ stands for $R_\infty^+$, $R_\infty^-$ depending on whether $N$
in $\nu_N$ is even or odd.
\end{lemma}
\begin{proof}
Let the kernel $W(x,y)$ of the integral operator $W$ be bounded by
\begin{equation*}
  |W(x,y)| \leq |W_1(x) f(x,y) W_2(x)|
\end{equation*}
where $W_1,W_2\in L^2(\R)$ and $f\in L^\infty(\R^2)$. By the Cauchy-Schwarz inequality
\begin{equation*}
\begin{split}
  \| W\varphi\|^2
      & \leq \int_{-L}^L |W_1(x)|^2 \int_{-L}^L |f(x,y)W_2(y)|^2\, dy\, dx\, \|\varphi\|^2\\
      & \leq \|f\|_\infty^2 \|W_1\|_2^2\, \|W_2\|_2^2\, \|\varphi\|^2
\end{split}
\end{equation*}
for $\varphi\in L^2[-L,L]$. The norms of $W_{1,2}$ and $f$ pertain to $\R$ and $\R^2$,
respectively. Hence
\begin{equation*}
  \|W\| \leq \|f\|_\infty \, \|W_1\|_2\, \|W_2\|_2 .
\end{equation*}
In order to prove \eqref{f_resolvent02t01} we can take $W_{1,2}=\sqrt{|V|}$ and
$f\equiv 1$ because of \eqref{f_resolvent01t02}. By the same estimate we can
prove \eqref{f_resolvent02t03} by using $W_1 = X\sqrt{|V|}$, $W_2=\sqrt{|V|}$,
$f\equiv 1$ and $W_1 = \sqrt{|V|}$, $W_2=X\sqrt{|V|}$, $f\equiv 1$.

Because of the obvious bound $|G(\nu_N;x,y)|\leq 1$
(cf. \eqref{f_resolvent_P}) we obtain
\begin{equation*}
  \|\sqrt{|V|}G(\nu_N)\sqrt{|V|}\| \leq V_L(0) ,
\end{equation*}
which proves \eqref{f_resolvent02t04} for $G(\nu_N)$ via \eqref{V_L_estimates}. 
Using this and \eqref{f_resolvent01bt02} along with $|\tau(Ls)|=1$ we obtain
for all $s\in\R$
\begin{equation*}
\begin{split}
  \lefteqn{2\sqrt{\nu_N} \|\sqrt{|V|}R_\infty^\pm(\nu_N,s)\sqrt{|V|}\|}\\
    & \leq \|\sqrt{|V|}P_s(\nu_N)\sqrt{|V|}\|
         +  \|\sqrt{|V|}P_c(\nu_N)\sqrt{|V|}\|
          + \|\sqrt{|V|}G(\nu_N)\sqrt{|V|}\|\\
    & \leq 3 V_L(0) ,
\end{split}
\end{equation*}
which gives \eqref{f_resolvent02t04} for $R_\infty^\pm(\nu_N,s)$ via \eqref{V_L_estimates}.

In order to prove \eqref{f_resolvent02t06} we use the kernel estimates in Lemma
\ref{f_resolvent01t} and obtain
\begin{align*}
  \| \sqrt{|V|}(P_{s,c}(z) - P_{s,c}(\nu_N))\sqrt{|V|} \| 
     & \leq 2 |s| V_L^{(2)}(2s)^{\frac{1}{2}}V_L(2|s|)^{\frac{1}{2}}, \\ 
  \| \sqrt{|V|}(G(z) - G(\nu_N))\sqrt{|V|} \|
     & \leq 2 |s| V_L^{(2)}(2|s|)^{\frac{1}{2}}V_L(2|s|)^{\frac{1}{2}},
\end{align*}
which would also be true for other real values than $\nu_N$ but that is not needed here.
Using \eqref{f_resolvent_L} and $|\tau(Ls)|=1$ we obtain
\begin{equation*}
\begin{split}
  \lefteqn{2\|(R_\infty(\nu_N,Ls) - R(z))\|}\\
    & \leq \frac{1}{\sqrt{\nu_N+s^2}} \big( 
          \| P_s(\nu_N) - P_s(z) \|
          + \| P_c(\nu_N) - P_c(z) \|
          + \| G(\nu_N) - G(z)\| \big) \\
    & \quad + \big| \frac{1}{\sqrt{\nu_N}} - \frac{1}{\sqrt{\nu_N}+is} \big|
        \big( \|P_s(\nu_N)\| + \|P_c(\nu_N)\| + \|G(\nu_N)\|\big) .
\end{split}
\end{equation*}
Here, $R_\infty$ means $R_\infty^+$ for even $N$ and $R_\infty^-$ for
odd $N$. This proves \eqref{f_resolvent02t06}.
\end{proof}

%% file: free_truncated_resolvent.tex
Let $S_N(z):=P_NR(z)$ be the truncated resolvent with the spectral projection
from \eqref{free_spectral_projection}. We need to control $S_N(z)$ on the
entire Fermi parabola $\Gamma_N$ (see \eqref{fermi_parabola}) and, with more
care, at the Fermi energy \eqref{fermi_energy}.

\begin{lemma}\label{f_truncated01t}
Let $V\in L^1(\R)$ and $z\in\Gamma_N$. Then, $S_N(z)=P_NR(z)$ satisfies
\begin{align}
  \big\|\sqrt{|V|}S_N(z)\sqrt{|V|}\big\| 
     & \leq \frac{8}{\pi} \|V\|_1 \frac{1}{\sqrt{\nu_N+s^2}}\ln(N+1) \label{f_truncated01t01}, \\
  \big\|\sqrt{|V|}(S_N(z) - S_N(\nu_N))\sqrt{|V|}\big\| 
     & \leq \frac{64}{\pi\nu_N} \|V\|_1 |s| (N+\frac{1}{2}) \label{f_truncated01t02}.
\end{align}
\end{lemma}
\begin{proof}
We start off from the spectral representation of $S_N$,
\begin{equation*}
  \sqrt{|V|} S_N(z)\sqrt{|V|} 
      = \sum_{j=1}^N \frac{1}{z-\lambda_j}
      (\sqrt{|V|}\varphi_j,\cdot)\sqrt{|V|}\varphi_j .
\end{equation*}
Applying the estimate \eqref{free_eigenfunction_delocalized} and then using
\eqref{fermi_parabola_distance} we obtain
\begin{equation}\label{f_truncated01t03}
\begin{split}
  \big\| \sqrt{|V|} S_N(z) \sqrt{|V|}\big\| 
     & \leq \frac{1}{L}\|V\|_1\sum_{j=1}^N \frac{1}{|z-\lambda_j|}\\
     & \leq \frac{2}{\pi}\|V\|_1 \frac{1}{\sqrt{\nu_N+s^2}}
           \sum_{j=1}^N \frac{1}{N+\frac{1}{2}-j} .
\end{split}
\end{equation}
Likewise, using 
\begin{equation*}
  \Big| \frac{z-\nu_N}{z-\lambda_j}\Big|^2
     = \frac{4\nu_N+s^2}{(\sqrt{\nu_N}+\sqrt{\lambda_j})^2+s^2}
       \frac{s^2}{(\sqrt{\nu_N}-\sqrt{\lambda_j})^2+s^2}
     \leq 4\frac{s^2}{(\sqrt{\nu_N}-\sqrt{\lambda_j})^2}
\end{equation*}
we find
\begin{multline}\label{f_truncated01t04}
  \|\sqrt{|V|}(S_N(z)-S_N(\nu_N))\sqrt{|V|}\| \\
       \leq \frac{1}{L}\|V\|_1
            \sum_{j=1}^N \frac{|z-\nu_N|}{|z-\lambda_j||\nu_N-\lambda_j|}
       \leq |s| \frac{2}{L}\|V\|_1
           \sum_{j=1}^N\frac{1}{(\sqrt{\nu_N}-\sqrt{\lambda_j})
                        (\nu_N-\lambda_j)} \\
       \leq 4|s| \frac{\|V\|_1}{\pi^3} \frac{(2L)^2}{N+\frac{1}{2}} 
           \sum_{j=1}^N\frac{1}{(N+\frac{1}{2}-j)^2} .
\end{multline}
Applying \eqref{estimates_sum} to the sums in \eqref{f_truncated01t03} and
\eqref{f_truncated01t04} we obtain \eqref{f_truncated01t01} and
\eqref{f_truncated01t02}.
\end{proof}

The asymptotic analysis in Section \ref{asymptotics_d} is based upon a formula
for the kernel of the truncated resolvent.

\begin{proposition}\label{f_truncated02t}
Let $z>0$ such that $\sqrt{z}>\frac{\pi}{2L}N$. Then, the kernel $S_N(z;x,y)$ of
the operator $S_N(z)=P_NR(z)$ decomposes into
\begin{equation*}
\begin{split}
  S_N(z;x,y) 
    & = \varkappa_N S_{0,N}(z;x,y) - S_{1,N}(z;x,y) \\
    & \quad - (-1)^N (\tilde\varkappa_N\tilde S_{0,N}(z;x,y) - \tilde S_{1,N}(z;x,y) )
\end{split}
\end{equation*}
with the constants
\begin{gather*}
  \varkappa_N := \int_0^\infty e^{-\frac{2L}{\pi}\sqrt{z}v}
              \frac{\sinh((N+\frac{1}{2})v)}{\sinh\frac{v}{2}}\, dv , \\
  \tilde\varkappa_N := \int_0^\infty e^{-\frac{2L}{\pi}\sqrt{z}v}
     \frac{\cosh((N+\frac{1}{2})v)}{\cosh\frac{v}{2}}\, dv
\end{gather*}
and the kernel functions
\begin{gather*}
    S_{0,N}(z;x,y) := \frac{\cos(\sqrt{z}(x-y))}{2\pi\sqrt{z}},\
    \tilde S_{0,N}(z;x,y) := \frac{\cos(\sqrt{z}(x+y))}{2\pi\sqrt{z}}, \\
\begin{aligned}
    S_{1,N}(z;x,y) & := \frac{1}{2\pi\sqrt{z}}\int_0^{\frac{\pi(x-y)}{2L}}
      \sin\big(\frac{2L}{\pi}\sqrt{z}(u-\frac{\pi(x-y)}{2L})\big)
       \frac{\sin((N+\frac{1}{2})u)}{\sin\frac{u}{2}}\, du  , \\
    \tilde S_{1,N}(z;x,y) & := \frac{1}{2\pi\sqrt{z}}\int_0^{\frac{\pi(x+y)}{2L}}
      \sin\big(\frac{2L}{\pi}\sqrt{z}(u-\frac{\pi(x+y)}{2L})\big)
       \frac{\cos((N+\frac{1}{2})u)}{\cos\frac{u}{2}}\, du .
\end{aligned}
\end{gather*}
\end{proposition}
\begin{proof}
With the eigenfunctions from \eqref{free_spectral_quantities} and using the
product formulae for sine and cosine we can write
\begin{equation*}
\begin{split}
  S_N(z;x,y)
   & = \sum_{j=1}^N\frac{1}{z-\lambda_j}\varphi_j(x)\bar\varphi_j(y)\\
   & =  \frac{1}{2L} \Big[
     \sum_{j=1}^N\frac{1}{z-\lambda_j}\cos\big(\frac{\pi j(x-y)}{2L}\big)
    -\sum_{j=1}^N\frac{(-1)^j}{z-\lambda_j}\cos\big(\frac{\pi j(x+y)}{2L}\big)
    \Big] .
\end{split}
\end{equation*}
In order to sum the series we
write the fraction as a Laplace transform. It is convenient to put
$z=\frac{\pi^2}{4L^2}\tilde z^2$. Then,
\begin{equation*}
  \frac{1}{z-\lambda_j}
    = \frac{4L^2}{\pi^2}\frac{1}{2\tilde z}\big(\frac{1}{\tilde z-j} 
          + \frac{1}{\tilde z+j}\big)
    = \frac{4L^2}{\pi^2}\frac{1}{\tilde z}
        \int_0^\infty e^{-\tilde z v}\cosh (jv)\, dv
\end{equation*}
since $\tilde z>j$ by assumption. Hence,
\begin{equation*}
\begin{split}
  \lefteqn{S_N(z;x,y)} \\ 
    & = \frac{2L}{\pi^2}\frac{1}{\tilde z}
   \int_0^\infty e^{-\tilde z v}\sum_{j=1}^N \cosh (jv)
    \Big[ \cos\big(\frac{\pi j(x-y)}{2L}\big) - (-1)^j\cos\big(\frac{\pi j(x+y)}{2L}\big)\Big] \,
     dv .
\end{split}
\end{equation*}
Using $\cos(\alpha)\cosh(\beta)=\re(\cos(\alpha+i\beta))$ for
$\alpha,\beta\in\R$ we obtain
\begin{equation*}
   S_N(z;x,y) = \frac{L}{\pi^2\tilde z} (\re I_s - (-1)^N \re I_c)
\end{equation*}
with the integrals
\begin{equation*}
   I_s := \int_0^\infty e^{-\tilde z v}
     \frac{\sin(M(a+iv))}{\sin(\frac{1}{2}(a+iv))}\,
      dv,\
   I_c := \int_0^\infty e^{-\tilde z v}
     \frac{\cos(M(b+iv))}{\cos(\frac{1}{2}(b+iv))}\,
      dv .
\end{equation*}
Here we abbreviated
\begin{equation*}
  M:=N+\frac{1}{2},\ a := \frac{\pi(x-y)}{2L},\ b := \frac{\pi(x+y)}{2L} .
\end{equation*}
We evaluate the integral $I_s$ by changing the integration contour. To this end, put
$w = a+iv$ and $\Gamma=\Gamma_1\cup\Gamma_2\cup\Gamma_3\cup\Gamma_4$,
\begin{gather*}
  \Gamma_1 = \{ a+iv\mid 0\leq v\leq R\},\
  \Gamma_2 = \{u+iR\mid 0\leq u\leq a\}, \\
  \Gamma_3 = \{ iv\mid 0\leq v\leq R\},\
  \Gamma_4 = \{u\mid 0\leq u\leq a\},
\end{gather*}
orientated counterclockwise. By Cauchy's integral theorem
\begin{equation*}
  \int_\Gamma e^{i\tilde z w}\frac{\sin (Mw)}{\sin\frac{w}{2}}\, dw = 0
\end{equation*}
since the integrand has only removable singularities. Because of
\begin{equation*}
  e^{i\tilde z w}\frac{\sin (Mw)}{\sin\frac{w}{2}} 
   = e^{i\tilde z(u+iR)}\frac{\sin (M(u+iR))}{\sin(\frac{1}{2}(u+iR))}
   \sim e^{-(\tilde z-M+\frac{1}{2})R},\ R\to\infty
\end{equation*}
and $\tilde z-M+\frac{1}{2}>0$ the integral over $\Gamma_2$ vanishes as $R\to\infty$. Hence,
\begin{equation*}
   e^{i\tilde z a} I_s
     = \int_0^\infty e^{-\tilde z v}\frac{\sin (Miv)}{\sin(\frac{1}{2}iv)}\, dv
      + i \int_0^a e^{i\tilde z u}\frac{\sin (Mu)}{\sin\frac{u}{2}}\, du
\end{equation*}
and furthermore
\begin{equation*}
  \re I_s
    = \cos (\tilde z a) \int_0^\infty e^{-\tilde z v}\frac{\sinh (Mv)}{\sinh\frac{v}{2}}\, dv
      - \int_0^a \sin (\tilde z(u-a))\frac{\sin (Mu)}{\sin\frac{u}{2}}\, du.
\end{equation*}
This gives the terms $S_{0,N}$ and $S_{1,N}$.
The integral $I_c$ can be treated in like manner and the proof is finished.
\end{proof}

Via elementary calculations one can obtain the bounds
\begin{gather}\label{f_truncated01}
   | S_{0,N}(z;,x,y) | \leq \frac{1}{2\pi\sqrt{z}},\
   | \tilde S_{0,N}(z;,x,y) | \leq \frac{1}{2\pi\sqrt{z}} \\ \label{f_truncated02}
  | S_{1,N}(z;x,y)| \leq \frac{1}{\sqrt{z}} \frac{N+\frac{1}{2}}{2L} |x-y|,\
  |\tilde S_{1,N}(z;x,y)| \leq \frac{1}{\sqrt{z}} \frac{N+\frac{1}{2}}{2L} |x+y| .
\end{gather}
for the above kernel functions. 
Thereby, Proposition \ref{f_truncated02t} helps to
separate the $x,y$ and $N$ dependence of $S_N(z;x,y)$ for special real values of $z$ including
the Fermi energy \eqref{fermi_energy} such that (see Lemma \ref{estimates01t})
\begin{equation*}
  S_N(\nu_N;x,y)\sim\frac{\varkappa_N}{2\pi\sqrt{\nu_N}}\cos(\sqrt{\nu_N}(x-y))
  \ \text{with}\ \varkappa_N\sim\ln N,\ N\to\infty.
\end{equation*}
By the addition theorem for the cosine this leading term can be written as
\begin{equation}\label{f_truncated03}
  S_{0,N}(\nu_N) = \frac{1}{2\pi\sqrt{\nu_N}} ( P_s(\nu_N) + P_c(\nu_N) )
\end{equation}
with the rank-one operators from \eqref{f_resolvent_P}.

%% file: free_delta.tex
The delta-term being non trivial reflects on an abstract level the boundary
conditions used in the definition of $H$, which make up the difference between
$H$ and $-\nabla^2$.

\begin{proposition}\label{f_delta01t}
For $z\in\C\setminus\sigma(H)$, $\varphi\in D(\nabla^2)$, and the
resolvent $R(z)$ of $H$ we have
\begin{equation}\label{f_delta01t01}
  (R(z)(z\id+\nabla^2)\varphi)(x) = \varphi(x) -
   \frac{\sin(\sqrt{z}(x+L))\varphi(L)
     -\sin(\sqrt{z}(x-L))\varphi(-L)}{\sin(2\sqrt{z}L)} .
\end{equation}
Furthermore, the delta-term $D(z)$ from \eqref{anderson02t02} reads
\begin{equation}\label{f_delta01t02}
   D(z) = \frac{L}{4}
   \Big( \frac{1}{\sin^2(\sqrt{z}L)} P_s(z) + \frac{1}{\cos^2(\sqrt{z}L)} P_c(z) \Big)
\end{equation}
with the rank-one operators $P_s(z)$ and $P_c(z)$ from \eqref{f_resolvent_P}.
\end{proposition}
\begin{proof}
(a)
In order to derive \eqref{f_delta01t01} we integrate by parts two times using the
Dirichlet boundary conditions, i.e. $R(z;x,\pm L)=0$, in the first step
\begin{equation*}
\begin{split}
    \lefteqn{(R(z)\nabla^2\varphi)(x)}\\
    & = \int_{-L}^L R(z;x,y)\varphi''(y)\, dy \\
    & = \left[ R(z;x,y)\varphi'(y)\right]_{-L}^L 
        - \int_{-L}^L\frac{\partial R(z;x,y)}{\partial y}\varphi'(y)\, dy\\
    & = -\left[\frac{\partial R(z;x,y)}{\partial y}\varphi(y)\right]_{-L}^L 
      + \int_{-L}^L\frac{\partial^2 R(z;x,y)}{\partial y^2}\varphi(y)\, dy\\
    & = \frac{\partial R(z;x,-L)}{\partial y}\varphi(-L) 
        -\frac{\partial R(z;x,L)}{\partial y}\varphi(L)
        - z\int_{-L}^LR(z;x,y)\varphi(y)\, dy + \varphi(x) .
\end{split}
\end{equation*}
From the explicit form \eqref{f_resolvent_green} of $R(z;x,y)$ we deduce
\begin{equation*}
  \frac{\partial}{\partial y}R(z;x,\pm L) 
   = \frac{\sin(\sqrt{z}(x\pm L))}{\sin(2\sqrt{z}L)} ,
\end{equation*}
which implies \eqref{f_delta01t01}.

(b)
For Formula \eqref{f_delta01t02} we use \eqref{anderson02t03} along with \eqref{f_delta01t01},
\begin{equation*}
\begin{split}
  0 & = (R(z)(z\id+\nabla^2)D(z))(x,y) \\
    & = D(z;x,y) - \frac{\sin(\sqrt{z}(x+L))}{2\sin(\sqrt{z}L)\cos(\sqrt{z}L)} D(z;L,y)\\
    & \quad + \frac{\sin(\sqrt{z}(x-L))}{2\sin(\sqrt{z}L)\cos(\sqrt{z}L)}
               D(z;-L,y) .
\end{split}
\end{equation*}
By the Dirichlet boundary conditions
\begin{equation*}
  (R(z)\nabla[\nabla,R(z)])(\pm L,y)
    = \int_{-L}^L R(z;\pm L,y')(\nabla[\nabla,R(z)])(y',y)\, dy' = 0 .
\end{equation*}
Hence, by definition \eqref{anderson02t02} and the explicit form \eqref{f_resolvent_green} 
of $R(z)$ we get
\begin{equation*}
  D(z;\pm L,y) 
    = \pm\frac{L}{2} \nabla_xR(z;\pm L,y)
    = \pm\frac{L}{2}\frac{\sin(\sqrt{z}(y\pm
      L))}{2\sin(\sqrt{z}L)\cos(\sqrt{z}L)} .
\end{equation*}
Putting everything together we obtain
\begin{equation*}
\begin{split}
  \lefteqn{D(z;x,y)}\\
   & = \frac{L}{8}
      \frac{\sin(\sqrt{z}(x+L))\sin(\sqrt{z}(y+L))
            +\sin(\sqrt{z}(x-L))\sin(\sqrt{z}(y-L))}{\sin^2(\sqrt{z}L)\cos^2(\sqrt{z}L)},
\end{split}
\end{equation*}
which implies the statement via the usual trigonometric formulae.
\end{proof}

%% file: perturbed_resolvent.tex
Since the perturbed operator enters only through the T-operator and the operator
$\Omega(z)$ (cf. \eqref{anderson_T-matrix} and \eqref{anderson_Omega})
we have a closer look at suchlike operators. Recall from Section
\ref{anderson} that we already know those operators to exist for
$z\notin\R$. What is new herein is that the bounds hold true uniformly
on the entire Fermi parabola $\Gamma_N$ including the Fermi energy $\nu_N$
(cf. \eqref{fermi_parabola}, \eqref{fermi_energy}).

\begin{lemma}\label{p_resolvent01t}
Let $V\in L^1(\R)$ and assume in addition
\begin{equation}\label{p_resolvent01t00}
  q_\Omega:= \frac{4}{\sqrt{\nu_N}}\|V\|_1 < 1 .
\end{equation}
Then, the operators $\Omega(z)$ exist for all $z\in\Gamma_N$ and are
uniformly bounded with
\begin{equation}\label{p_resolvent01t01}
  \|\Omega(z)\| \leq \frac{1}{1-q_\Omega} =: C_\Omega .
\end{equation}
In particular, $\nu_N\notin\sigma(H_V)$. If in addition $V\in L^\infty(\R)$ then
the T-operator, $T(z)$ (see \eqref{anderson_T-matrix}), exists and is bounded
with $\|T(z)\| \leq C_\Omega\|V\|_\infty$.
\end{lemma}
\begin{proof}
Using $\|J\|=1$ we obtain from \eqref{f_resolvent02t01} the bound
\begin{equation*}
  \|\sqrt{|V|}R(z)\sqrt{|V|}J\|
     \leq \frac{4}{\sqrt{\nu_N}} \|V\|_1
     < 1 .
\end{equation*}
Hence, a Neumann series argument shows that
$\Omega(z)$ exists and is bounded with \eqref{p_resolvent01t01}. Since
$\nu_N\notin\sigma(H)$ by construction the remark
after \eqref{anderson_Omega_T} shows $\nu_N\notin\sigma(H_V)$.
Furthermore,
\begin{equation*}
  \| T(z) \| = \| \sqrt{|V|} J \Omega(z) \sqrt{|V|} \| 
             \leq \|\sqrt{|V|}\|_\infty^2 \|J\| \|\Omega(z)\|
             = \|V\|_\infty \|\Omega(z)\|
\end{equation*}
completes the proof.
\end{proof}

We had seen in Section \ref{f_resolvent} that it is advantageous to work with
the operators $R_\infty^\pm(\nu_N,s)$ (cf. \eqref{f_resolvent_limit})
instead of the resolvent $R(z)$. Likewise, we employ the operators
$\Omega_\infty^\pm(\nu_N,s)$,
\begin{equation}\label{p_resolvent_limit}
  \Omega_\infty^\pm(\nu_N,s) := (\id-\sqrt{|V|}R_\infty^\pm(\nu_N,s)\sqrt{|V|}J)^{-1}
\end{equation}
instead of $\Omega(z)$. In view of the rank-two operator in
\eqref{f_resolvent_limit} it is reasonable to define (cf. \eqref{f_resolvent_P})
\begin{equation}\label{p_resolvent_limit_Phi}
  \Phi(\nu_N) := (\id-\sqrt{|V|}K(\nu_N)\sqrt{|V|}J)^{-1},\
   K(\nu_N) := \frac{1}{2\sqrt{\nu_N}}G(\nu_N) .
\end{equation}
The operator $K(\nu_N)$ is
closely related to the resolvent of the free Schr\"odinger operator defined on the whole
of $\R$. That is why it replaces $G(\nu_N)$.

\begin{lemma}\label{p_resolvent02t}
Let $V\in L^1(\R)$. If
\begin{equation}\label{p_resolvent02t00}
  q_\infty:=\frac{3}{2\sqrt{\nu_N}}\|V\|_1 < 1,\
  q_\Phi:= \frac{1}{2\sqrt{\nu_N}}\|V\|_1 < 1,
\end{equation}
then the operators $\Omega_\infty^\pm(\nu_N,s)$ and $\Phi(\nu_N)$
(defined in \eqref{p_resolvent_limit}, \eqref{p_resolvent_limit_Phi}) exist and are bounded with
\begin{equation}\label{p_resolvent02t01}
  \|\Omega_\infty^\pm(\nu_N,s)\| \leq \frac{1}{1-q_\infty}=: C_{\Omega_\infty},
  \|\Phi(\nu_N)\| \leq \frac{1}{1-q_\Phi}=: C_{\Phi} .
\end{equation}
Furthermore, let $z\in\Gamma_N$. Then, (see \eqref{anderson_Omega})
\begin{equation}\label{p_resolvent02t02}
  \|\Omega(z)-\Omega_\infty(\nu_N,Ls)\| 
    \leq  C_{\Omega_\infty}' |s| (V_L(0)+ V_L^{(2)}(2|s|)^{\frac{1}{2}}V_L(2|s|)^{\frac{1}{2}} ),
\end{equation}
where $\Omega_\infty$ stands for $\Omega_\infty^+$, $\Omega_\infty^-$ depending
on whether $N$ in $\nu_N$ is even or odd.
The constant is (see \eqref{p_resolvent01t01})
\begin{equation*}
  C_{\Omega_\infty}' 
   :=  C_\Omega C_{\Omega_\infty} \frac{3}{\sqrt{\nu_N}} \max\{
   \frac{1}{2\sqrt{\nu_N}},1\} .
\end{equation*}
\end{lemma}
\begin{proof}
We know from \eqref{f_resolvent02t04} and the assumption \eqref{p_resolvent02t00} that
\begin{equation*}
    \|\sqrt{|V|}R_\infty^\pm(\nu_N,s)\sqrt{|V|}\| 
       \leq \frac{3}{2\sqrt{\nu_N}}\|V\|_1 < 1 .
\end{equation*}
A Neumann series argument shows that $\Omega_\infty^\pm(\nu_N,s)$ exists and is
bounded with \eqref{p_resolvent02t01}. The operator $\Phi(\nu_N)$ is treated in
like manner. For \eqref{p_resolvent02t02} note
\begin{equation*}
  \Omega(z)-\Omega_\infty(\nu_N,Ls) = \Omega(z)
  \sqrt{|V|}(R_\infty(\nu_N,Ls)-R(z))\sqrt{|V|} \Omega_\infty(\nu_N,s)
\end{equation*}
where $\Omega_\infty$ is $\Omega_\infty^+$ for even $N$ in $\nu_N$ and
$\Omega_\infty^-$ for odd $N$ and $R_\infty$ likewise.
Hence, \eqref{p_resolvent01t01}, \eqref{f_resolvent02t06}, and the first part \eqref{p_resolvent02t01}
conclude the proof.
\end{proof}

We only need the matrix elements with respect to the functions
$\omega_{s,c}(\nu_N)$ from \eqref{f_resolvent_omega}.
To this end, we introduce the $2\times 2$ matrices
\begin{equation}\label{omega_hat1}
\begin{split}
  \lefteqn{\hat\Omega_\infty^\pm(\nu_N,s)}\\ 
   & :=
\begin{pmatrix}
  (\omega_s(\nu_N), J\Omega_\infty^\pm(\nu_N,s)\omega_s(\nu_N))
    & (\omega_s(\nu_N), J\Omega_\infty^\pm(\nu_N,s)\omega_c(\nu_N)) \\
  (\omega_c(\nu_N), J\Omega_\infty^\pm(\nu_N,s)\omega_s(\nu_N))
    & (\omega_c(\nu_N), J\Omega_\infty^\pm(\nu_N,s)\omega_c(\nu_N))
\end{pmatrix}
\end{split}
\end{equation}
and
\begin{equation}\label{omega_hat2}
  \hat\Phi(\nu_N):=
\begin{pmatrix}
     (\omega_s(\nu_N),J\Phi(\nu_N)\omega_s(\nu_N))
    & (\omega_s(\nu_N),J\Phi(\nu_N)\omega_c(\nu_N)) \\
     (\omega_c(\nu_N),J\Phi(\nu_N)\omega_s(\nu_N))
    & (\omega_c(\nu_N),J\Phi(\nu_N)\omega_c(\nu_N))
\end{pmatrix} .
\end{equation}
Note that $\hat\Phi(\nu_N)^*=\hat\Phi(\nu_N)$ since $J\Phi(\nu_N)$ 
is self-adjoint. In this one-dimen\-sional case the above $2\times 2$ matrices
correspond to the eigenspace decomposition of angular momentum in higher
dimensions.

\begin{lemma}\label{p_resolvent03t}
Let $V\in L^1(\R)$ satisfy in addition (see \eqref{p_resolvent02t01})
\begin{equation}\label{p_resolvent03t01}
  \frac{1}{\sqrt{\nu_N}} \|V\|_1  C_{\Phi} < 1.
\end{equation}
Then, the $2\times 2$ matrices $\hat Z^\pm(\nu_N,s)$,
\begin{equation}\label{p_resolvent03t02}
   \hat Z^\pm(\nu_N,s) :=
      (\id \mp \frac{1}{2\sqrt{\nu_N}}\hat\Phi(\nu_N)\hat\tau(\pm s) )^{-1},\ s\in\R,\
\end{equation}
exist. Here (for $\tau$ see \eqref{f_resolvent_tau}) for $s\in\R$,
\begin{equation}\label{p_resolvent03t03}
\begin{gathered}
  \hat\tau(s) := \diag(\tau(s),-\tau(-s)),\   
  \lim_{s\to\infty}\hat\tau(s) = -i\id \\
  \hat\tau(s)^*\hat\tau(s)=\id,\
  \hat\tau(s)^*=\hat\tau(-s) .
\end{gathered}
\end{equation}
Furthermore, we have
\begin{equation}\label{p_resolvent03t04}
  \hat\Omega^\pm_\infty(\nu_N,s)
                 = \hat Z^\pm(\nu_N,s)\hat\Phi(\nu_N) 
                 = \hat\Phi(\nu_N)\hat Z^\pm(\nu_N,-s)^* .
\end{equation}
\end{lemma}
\begin{proof}
The operators $\Omega^\pm_\infty(\nu_N,s)$ have the form $(A - a_1(f_1,\cdot)g_1 - a_2(f_2,\cdot)g_2)^{-1}$
with an invertible operator $A$, vectors $f_{1,2}$, $g_{1,2}$, and
$a_{1,2}\in\C$. Computing the inverse on the vectors $g_{1,2}$ amounts to
solving the equations
\begin{equation*}
  (   A - a_1(f_1,\cdot)g_1 - a_2(f_2,\cdot)g_2 ) h_k = g_k, \ k=1,2,
\end{equation*}
for $h_{1,2}$. The matrix elements $(f_j,h_k)$, in particular, satisfy
\begin{equation*}
  (f_j,h_k) - a_1 (f_1,h_k)(f_j,A^{-1}g_1) - a_2(f_2,h_k)(f_j,A^{-1}g_2) = 
      (f_j,A^{-1}g_k)
\end{equation*}
 for $ j,k=1,2$. Introducing the $2\times 2$-matrices
\begin{equation*}
  \hat B := ( (f_j,h_k) )_{j,k=1,2} ,\ 
  \hat A := (f_j,A^{-1}g_k)_{j,k=1,2} ,\ 
  \hat a := \diag(a_1,a_2)
\end{equation*}
we can write this as
\begin{equation*}
  \hat B - \hat A \hat a \hat B = \hat A
\end{equation*}
which can easily be solved for $\hat B$. Now, for $\Omega^+(\nu,s)$ put
\begin{equation*}
  a_1 = \frac{\tau(s)}{2\sqrt{\nu_N}},\ a_2 = -\frac{\tau(-s)}{2\sqrt{\nu_N}},\
  f_{1,2} = J^*\omega_{s,c}(\nu_N),\ 
  g_{1,2} =\omega_{s,c}(\nu_N),
\end{equation*}
to obtain the first equality in \eqref{p_resolvent03t04}. The second follows
from
\begin{equation*}
\begin{split}
  \hat\Phi(\nu_N)\hat Z^+(\nu_N,-s)^*
    & = \hat\Phi(\nu_N) ( \id-\frac{1}{2\sqrt{\nu_N}} \hat\tau(-s)^* \hat\Phi(\nu_N))^{-1}\\
    & = (\id-\frac{1}{2\sqrt{\nu_N}}\hat\Phi(\nu_N)\hat\tau(s) )^{-1}\hat\Phi(\nu_N)
\end{split}
\end{equation*}
where we used the next to last relation in \eqref{p_resolvent03t03}. The
relations for $\hat\tau(s)$ are obvious.
In order to show that $\hat Z^+(\nu_N,s)$ is well-defined we look at the entries
\begin{equation*}
  | (\omega_{s,c}(\nu_N),J\Phi(\nu_N)\omega_{s,c}(\nu_N))|
     \leq \|V\|_1 \|\Phi(\nu_N)\| .
\end{equation*}
With the maximum norm $\|\cdot\|_\infty$ for matrices we thus get
\begin{equation*}
  \frac{1}{2\sqrt{\nu_N}}\|\hat\Phi(\nu_N)\hat\tau(s)\|_\infty
    \leq \frac{1}{\sqrt{\nu_N}} \|V\|_1 \|\Phi(\nu_N)\|
    \leq \frac{1}{\sqrt{\nu_N}} \|V\|_1 C_{\Phi}
    <1 .
\end{equation*}
Now a Neumann series argument proves the statement.
The matrix $\hat Z^-(\nu_N,s)$ is treated likewise.
\end{proof}

The Neumann series was the only abstract tool we used in proving invertibility
of operators. Therefore, the conditions put on the potential $V$ might be too
restrictive. For example, the operator
$\id-VR(z)$ is known to be invertible for all $z\in\C\setminus M$ with $M$ being a discrete
set (see \cite{ReedSimon1979}, p.114). 
Thus, more advanced tools could possibly help to allow for larger
classes of potentials. But that is not our main concern here.

%% file: perturbed_eigenvalues.tex
One important consequence of Lemma \ref{p_resolvent01t} is that the spectrum
$\sigma(H_V)$ of the operator $H_V$ on $L^2[-L,L]$ can be decomposed with
respect to the Fermi energy $\nu_N$,
\begin{equation}\label{p_decomposition}
\begin{gathered}
  \sigma(H_V) = \sigma_1(H_V) \cup \sigma_2(H_V),\\
  \sigma_1(H_V) := \{ \mu_j \mid \mu_j < \nu_N \},\
  \sigma_2(H_V) := \{ \mu_j \mid \mu_j > \nu_N \}.
\end{gathered}
\end{equation}
Equivalently, there is an $M=M(N)$ with
\begin{equation}\label{p_decomposition_M}
  \mu_j < \nu_N,\ j=1,\ldots,M,\
  \mu_j > \nu_N,\ j\geq M+1.
\end{equation}
Exactly $N$ free eigenvalues lie below $\nu_N$. We need to know how many
perturbed eigenvalues do so which amounts to estimating $M$. For the upper bound
we modify Bargmann's inequality on negative eigenvalues
(cf. \cite[Thm. XIII.9]{ReedSimon1978}).

\begin{proposition}\label{p_eigenvalues01t}
Let $V_-:=\min\{V,0\}$ satisfy
\begin{equation}\label{p_eigenvalues01t01}
  |V_-(x)| \leq \frac{C_\alpha}{(1+|x|)^{\alpha+1}}
\end{equation}
with $\alpha>0$ and $C_\alpha\geq 0$. Then, for all $E>0$
\begin{equation}\label{p_eigenvalues01t02}
\begin{gathered}
  M:= \#\{ \mu_j \mid \mu_j < E \} 
   \leq \frac{2L}{\pi} \sqrt{E} + C_E, \\
  C_E := \frac{1}{2E}\Big( 
          \frac{2C_\alpha}{\alpha\pi} (\|V_-\|_\infty+E)^{\frac{1}{2}} + \|V_-\|_\infty
          \Big) .
\end{gathered}
\end{equation}
In particular, with $E=\nu_N$ being the Fermi energy the bound becomes
\begin{equation}\label{p_eigenvalues01t03}
  M \leq  N + \frac{1}{2} + C_{\nu_N}.
\end{equation}
\end{proposition}
\begin{proof}
By the variational principle, the number of eigenvalues $M=M(V)$ satisfies
$M(V)\leq M(V_-)$. We may therefore assume that $V\leq 0$ and hence $V=-|V|$.
By a shift of the spectrum $M$ equals the number of negative eigenvalues of
\begin{equation*}
  -\psi'' - (|V|+E)\psi = \tilde\mu \psi,\ \psi(-L)=0=\psi(L) .
\end{equation*}
The eigenfunction $\psi_M$ corresponding to $\tilde\mu_M$ has exactly $M+1$
roots,
\begin{equation*}
  -L \leq x_0 < x_1 < \cdots < x_M = L .
\end{equation*}
Let us abbreviate
\begin{equation*}
  I_k := [x_k,x_{k+1}],\ V_k := \sup_{x\in I_k} |V(x)|,\ k=0,\ldots, M-1 .
\end{equation*}
Apparently, $\tilde\mu_M$ is a negative eigenvalue for the Dirichlet problem on
each $I_k$. We want a lower bound for the distance of two consecutive
roots. To this end, we estimate
\begin{equation*}
\begin{split}
  \int_{I_k} ( |V(x)| + E )|\psi_M(x)|^2 \, dx
      & \leq (V_k + E ) \int_{I_k} |\psi_M(x)|^2\, dx \\
      & \leq (V_k + E ) \big( \frac{x_{k+1}-x_k}{\pi} \big)^2 
                  \int_{I_k} |\psi_M'(x)|^2\, dx ,
\end{split}
\end{equation*}
where we used Wirtinger's inequality (see \cite{DymMcKean1985}) or in other words, the variational
principle for the lowest Dirichlet eigenvalue. If we had
\begin{equation*}
  (V_k + E ) \big( \frac{x_{k+1}-x_k}{\pi} \big)^2 \leq 1
\end{equation*}
the differential equation and the Dirichlet conditions would imply
\begin{multline*}
  \int_{I_k} (|V(x)|+E) |\psi_M(x)|^2\, dx\\
     \leq \int_{I_k} |\psi_M'(x)|^2\, dx
       =  \int_{I_k} (|V(x)|+E) |\psi(x)|^2 \, dx
            + \tilde\mu_M \int_{I_k} |\psi_M(x)|^2\, dx .
\end{multline*}
This is impossible for $\tilde\mu_M<0$ and thus
\begin{equation}\label{p_eigenvalues01t04}
  1 \leq \frac{(x_{k+1}-x_k)^2}{\pi^2} (V_k + E) .
\end{equation}
Since $V_k\leq \|V\|_\infty$ we obtain a first rough but uniform bound
\begin{equation*}
  x_{k+1} - x_k \geq \frac{\pi}{(\|V\|_\infty + E)^{\frac{1}{2}}} =:\delta,\
  0\leq k\leq M-1 .
\end{equation*}
Now we estimate in \eqref{p_eigenvalues01t04},
\begin{equation*}
  \pi 
    \leq (x_{k+1}-x_k) ( V_k + E )^{\frac{1}{2}}
    \leq (x_{k+1}-x_k) \sqrt{E} 
         \big( 1 + \frac{V_k}{2 E}\big) .
\end{equation*}
This can be cast into the form
\begin{equation*}
  1 \leq \frac{\sqrt{E}}{\pi} (x_{k+1} - x_k)
            + \frac{V_k}{2E} \frac{1}{1+\frac{V_k}{2E}}
    \leq \frac{\sqrt{E}}{\pi} (x_{k+1} - x_k) + \frac{V_k}{2E}
\end{equation*}
Summing up from $0$ to $M-1$ we obtain
\begin{equation*}
  M \leq \frac{2L}{\pi}\sqrt{E} + \frac{1}{2E}\sum_{k=0}^{M-1} V_k .
\end{equation*}
Using \eqref{p_eigenvalues01t01} we compare the sum with the integral of the
majorant of $V$
\begin{equation*}
\begin{split}
  \sum_{k=0}^{M-1} V_k
    & \leq \frac{1}{\delta}\sum_{\substack{0\leq k\leq M-2\\ x_{k+2}\leq 0}} (x_{k+2}-x_{k+1}) V_k
       + \frac{1}{\delta}\sum_{\substack{1\leq k\leq M-1\\ x_{k-1}>0}} (x_k-x_{k-1}) V_k
       + \|V\|_\infty \\
    & \leq \frac{1}{\delta}\int_{-L}^L \frac{C_\alpha}{(1+|x|)^{\alpha+1}} \, dx
       + \|V\|_\infty,
\end{split}
\end{equation*}
where $\|V\|_\infty$ is due to the summand that was left out.
This proves \eqref{p_eigenvalues01t02}. Finally, \eqref{p_eigenvalues01t03} is
an immediate consequence of the definition \eqref{fermi_energy} of $\nu_N$.
\end{proof}

An upper bound on the eigenvalues gives a lower bound on their number.

\begin{proposition}\label{p_eigenvalues02t}
Let $V_+:=\max\{V,0\}\in L^1(\R)$. Then, the perturbed eigenvalues satisfy
\begin{equation}\label{p_eigenvalues02t01}
  \sqrt{\mu_k} \leq \frac{k\pi}{2L} + \frac{1}{k\pi}\|V_+\|_1 .
\end{equation}
Moreover, for $E>0$ satisfying
\begin{equation}\label{p_eigenvalues02t02}
  E \geq \frac{2}{L}\|V_+\|_1
\end{equation}
the number of eigenvalues below $E$ has the lower bound
\begin{equation}\label{p_eigenvalues02t03}
  M:= \#\{ \mu_k \mid \mu_k\leq E \}
     \geq \frac{2L}{\pi} \sqrt{E} - \frac{2\|V_+\|_1}{\pi} \frac{1}{\sqrt{E}} -1 .
\end{equation}
In particular, with $E=\nu_N$ being the Fermi energy this becomes
\begin{equation}\label{p_eigenvalues02t03a}
  M \geq N-\frac{1}{2} - \frac{2\|V_+\|_1}{\pi} \frac{1}{\sqrt{\nu_N}} .
\end{equation}
\end{proposition}
\begin{proof}
By the variational principle, the eigenvalues $\mu_j=\mu_j(V)$ and the number of
eigenvalues $M=M(V)$ satisfy $\mu_j(V)\leq \mu_j(V_+)$ and $M(V)\geq M(V_+)$.
Thus, we may assume $V\geq 0$.
In \eqref{perturbed_eigenvalue_problem} we use the modified Pr\"ufer variables
\begin{equation*}
  0 \neq
\begin{pmatrix}
  \frac{1}{\sqrt{\mu}} \psi' \\ \psi
\end{pmatrix}
  = r
\begin{pmatrix}
  \cos\vartheta \\ \sin\vartheta
\end{pmatrix} .
\end{equation*}
The phase function $\vartheta$ satisfies the initial value problem
\begin{equation}\label{p_eigenvalues02t04}
  \vartheta' = \sqrt{\mu} - \frac{V}{\sqrt{\mu}}\sin^2\vartheta,\
  \vartheta(-L,\mu) = 0 .
\end{equation}
Integrating yields
\begin{equation}\label{p_eigenvalues02t05}
  \vartheta(L,\mu) = 2L\sqrt{\mu} - 
                 \frac{1}{\sqrt{\mu}}\int_{-L}^L V(y)\sin^2\vartheta(y,\mu)\, dy .
\end{equation}
To give a solution of \eqref{perturbed_eigenvalue_problem} is equivalent to
$\vartheta(L,\mu_k) = k\pi$, $k\in\N$. We show that $\vartheta(x,\mu)$
is strictly increasing in $\mu$ or more precisely
\begin{equation*}
  \Theta(x,\mu) := \frac{\partial}{\partial\mu}\vartheta(x,\mu) > 0 .
\end{equation*}
From \eqref{p_eigenvalues02t05} we deduce
\begin{equation*}
  \Theta' = \frac{1}{2\sqrt{\mu}}( 1+\frac{V}{\mu}\sin^2\vartheta)
               - \frac{V}{\sqrt{\mu}}\sin(2\vartheta)\Theta,\ \Theta(-L,\mu)=0 .
\end{equation*}
With the abbreviation
\begin{equation*}
  a(x) := -\frac{1}{\sqrt{\mu}}\int_{-L}^x V(y)\sin2\vartheta(y)\, dy
\end{equation*}
we obtain
\begin{equation*}
  \Theta(x,\mu) = \frac{1}{2\sqrt{\mu}} e^{a(x)}\int_{-L}^x e^{-a(y)}
   \Big[1+\frac{V(y)}{\mu}\sin^2\vartheta(y)\Big]\, dy >0 .
\end{equation*}
Furthermore, from \eqref{p_eigenvalues02t05} it is obvious that
\begin{equation*}
  \limsup_{\mu\to +0} \vartheta(L,\mu)\leq 0,\ \liminf_{\mu\to+\infty} \vartheta(L,\mu)
  = \infty .
\end{equation*}
We conclude that $\mu_k$ is the unique solution of the eigenvalue equation
\begin{equation*}
    k\pi = 2L\sqrt{\mu} - \frac{1}{\sqrt{\mu}}\int_{-L}^L
    V(y)\sin^2\vartheta(y)\, dy .
\end{equation*}
This implies the bound \eqref{p_eigenvalues02t01}
since $\mu_k\geq\lambda_k$. A lower bound for $M$ is thus given by the
largest $k$ such that
\begin{equation*}
  \frac{k\pi}{2L} + \frac{\|V\|_1}{k\pi} \leq \sqrt{E}
\end{equation*}
which can be written equivalently
\begin{equation*}
  \Big( k-\frac{L\sqrt{E}}{\pi}\Big)^2 
    \leq \frac{L^2E }{\pi^2} - \frac{2L}{\pi^2} \|V\|_1 =:r_{E,L}^2 .
\end{equation*}
The righthand side is positive by \eqref{p_eigenvalues02t02}. Solving
for $k$ yields two inequalities
\begin{equation*}
  \frac{2L\|V\|_1}{\pi^2} \frac{1}{ \frac{L\sqrt{E}}{\pi} + r_{E,L}}
    \leq k \leq 
    \frac{2L\sqrt{E}}{\pi} -
    \frac{2L\|V\|_1}{\pi^2}\frac{1}{ \frac{L\sqrt{E}}{\pi} + r_{E,L}} .
\end{equation*}
These are surely satisfied when
\begin{equation*}
  \frac{2\|V\|_1}{\pi} \frac{1}{\sqrt{E}}
     \leq k \leq
  \frac{2L}{\pi}\sqrt{E} - \frac{2\|V\|_1}{\pi} \frac{1}{\sqrt{E}} 
\end{equation*}
which makes sense because of \eqref{p_eigenvalues02t02}. The
righthand side differs from the next smaller integer by at most one which
proves \eqref{p_eigenvalues02t03}.
\end{proof}

%% file: delta_estimate.tex
An integral containing Dirac's delta function reduces to a point evaluation of the
integrand. A similar effect will be employed in Proposition \ref{asymptotics_d01t}.
The necessary estimates are dubbed delta estimates for that reason.
To any bounded function $f:\R^+\to\R^+$ we associate the transformed function
\begin{equation}\label{L_transform}
  f^*(L) := \int_0^\infty \frac{e^{-Ls}}{\sqrt{a+s^2}} f(s) 
              \int_{-L}^L e^{s|x|}\, dx\, ds,\ a>0 .
\end{equation}
The inner integral is motivated by the estimate \eqref{V_L_estimates}.

\begin{lemma}\label{delta01t}
Let $W\in L^1(\R)$ satisfy $X^nW\in L^\infty(\R)$ with some $n\in\N_0$ and
define $W_L$ as in \eqref{V_L}.
Let $g\geq 0$ be bounded and weakly differentiable with $g'\leq 0$. Then,
\begin{equation}\label{delta01t01}
  \int_0^\infty \frac{e^{-Ls}}{\sqrt{a+s^2}} g(s) W_L^{(m)}(s)\, ds
      \leq L^{m-1}g(0) n \|W\|_1 + L^{m-n} g^*(L) \|X^nW\|_\infty
\end{equation}
for all $m\in\N_0$.
Moreover, let $h\geq 0$ be bounded and weakly differentiable with $h(0)=0$ and
$h'\leq g$. For all $m\in\N_0$,
\begin{equation}\label{delta01t02}
  \int_0^\infty \frac{e^{-Ls}}{\sqrt{a+s^2}} h(s) W_L(s)\, ds
    \leq \frac{n^2}{L^2} g(0) \|W\|_1 + 
         \frac{1}{L^n} \Big[ \frac{n}{L} g^*(L) + h^*(L) \Big] \|X^nW\|_\infty .
\end{equation}
\end{lemma}
\begin{proof}
Let $f\geq 0$ be weakly differentiable and bounded. 
Integration by parts and dropping the negative term that appears yields the
following inequality
\begin{equation}\label{delta01t03}
\begin{split}
  L \int_0^\infty \frac{e^{-Ls}}{\sqrt{a+s^2}} f(s) W_L^{(p)}(s) \, ds
    &  \leq f(0) W_L^{(p)}(0)
       +  \int_0^\infty \frac{e^{-Ls}}{\sqrt{a+s^2}} f'(s)W_L^{(p)}(s)\, ds\\
    & \quad  +  \int_0^\infty \frac{e^{-Ls}}{\sqrt{a+s^2}} f(s) 
               W_L^{(p+1)}(s)\, ds .
\end{split}
\end{equation}
(a) When $f=g$ in \eqref{delta01t03} the integral containing $g'$ becomes
non-positive and can be dropped. Iterating the resulting inequality $n$-times
yields
\begin{multline*}
  L\int_0^\infty \frac{e^{-Ls}}{\sqrt{a+s^2}} g(s) W_L^{(m)}(s)\, ds \\
      \leq  g(0) \sum_{k=0}^{n-1} \frac{1}{L^k} W_L^{(m+k)}(0)
        + \int_0^\infty \frac{e^{-Ls}}{\sqrt{a+s^2}} g(s)
                W_L^{(m+n)}(s)\, ds .
\end{multline*}
Using the estimates \eqref{V_L_estimates} with $p=m+k$, $q=0$ in the sum and
$p=m$, $q=n$ in the integral we obtain \eqref{delta01t01}.

(b) With $f=h$ in \eqref{delta01t03} we get
\begin{multline*}
  \int_0^\infty \frac{e^{-Ls}}{\sqrt{a+s^2}} h(s) W_L^{(k)}(s) \, ds \\
     \leq \frac{1}{L} \int_0^\infty \frac{e^{-Ls}}{\sqrt{a+s^2}} g(s)
          W_L^{(k)}(s)\, ds
     + \frac{1}{L} \int_0^\infty \frac{e^{-Ls}}{\sqrt{a+s^2}} h(s)
                 W_L^{(k+1)}(s)\, ds .
\end{multline*}
After iterating we obtain
\begin{equation*}
\begin{split}
  \lefteqn{\int_0^\infty \frac{e^{-Ls}}{\sqrt{a+s^2}} h(s) W_L(s) \, ds}\\
    & \leq \frac{1}{L} \sum_{k=0}^{n-1} 
           \frac{1}{L^k}\int_0^\infty \frac{e^{-Ls}}{\sqrt{a+s^2}} g(s)
              W_L^{(k)}(s)\, ds
         + \frac{1}{L^n}\int_0^\infty \frac{e^{-Ls}}{\sqrt{a+s^2}} h(s)
              W_L^{(n)}(s)\, ds \\
    & \leq \frac{1}{L} \sum_{k=0}^{n-1} \frac{1}{L^k} \big[ 
             nL^{k-1}\|W\|_1 g(0) + L^{k-n} g^*(L) \|X^nW\|_\infty \big]
         + \frac{1}{L^n} h^*(L) \|X^nW\|_\infty
\end{split}
\end{equation*}
where we estimated the integrals in the sum via (a) and the remaining integral
by \eqref{V_L_estimates} with $p=n$, $q=0$. That concludes the proof.
\end{proof}

We can now formulate the delta estimate.

\begin{proposition}\label{delta02t}
Let $W\in L^1(\R)$ satisfy $X^nW\in L^\infty(\R)$ for some $n\in\N_0$ and define
$W_L$ as in \eqref{V_L}. 
Assume that $f_L:\R^+\to\R^+$ obey
\begin{equation}\label{delta_rate_functions}
  f_L(s) \leq s \Theta(L) \ \text{and}\
  f_L(s) \leq \vartheta(L),\ L>0,\ s\in\R^+ ,
\end{equation}
with functions $\vartheta,\Theta:\R^+\to\R^+$. Then,
\begin{multline}\label{delta02t01}
  \int_0^\infty \frac{e^{-Ls}}{\sqrt{a+s^2}} f_L(s) W_L(s) \, ds
    \leq \frac{n^2\Theta(L)}{L^2} \|W\|_1 \\
          + \frac{2}{L^n}\Big[ \frac{n\Theta(L)}{L}( 1+ \sqrt{2}\ln(L+1))
          + \frac{1}{2}\frac{L\vartheta(L)^2}{\Theta(L)}
          + \sqrt{2}\ln\big( \frac{\Theta(L)}{\vartheta(L)}+1\big) \Big] \|X^nW\|_\infty .
\end{multline}
\end{proposition}
\begin{proof}
We want to apply \eqref{delta01t02} in Lemma \ref{delta01t}. To this end, we define
$g_L(s):=\Theta(L)$ and
\begin{equation*}
  h_L(s) :=
\begin{cases}
  s \Theta(L) & \text{for}\ s \leq \eta,  \\
  \vartheta(L) & \text{for}\ s \geq \eta,
\end{cases}
 \
  h_L'(s) =
\begin{cases}
   \Theta(L) & \text{for}\ s < \eta,  \\
   0 & \text{for}\ s > \eta,
\end{cases}
\end{equation*}
where $\eta:=\vartheta(L)/\Theta(L)$ for short. Obviously, $f_L\leq h_L$
and $h_L'\leq g_L$. Thus, we only need to estimate $g_L^*$ and $h_L^*$.
For any admissible $f$ (e.g. bounded)
\begin{equation*}
\begin{split}
  \frac{1}{2}f^*(L)
    & = \int_0^\infty \frac{e^{-Ls}}{\sqrt{a+s^2}}\frac{1}{s}
       f(s)\int_0^{sL}e^{-x} \, dx \, ds \\
    &  \leq L \int_0^a \frac{f(s)}{\sqrt{a+s^2}} \, ds
            + \int_a^\infty \frac{f(s)}{\sqrt{a+s^2}s}\, ds
\end{split}
\end{equation*}
with some $a>0$. For $f=g_L$ we choose $a:=1/L$ and substitute $s\mapsto 1/s$ in the second integral. Then,
\begin{equation*}
  \frac{1}{2}g_L^*(L) 
      \leq \Theta(L) \big[ L\int_0^{\frac{1}{L}} \, ds + \sqrt{2}\int_0^L \frac{1}{1+s}\, ds \  \big]
      = \Theta(L)(1 + \sqrt{2}\ln(1+L)) .
\end{equation*}
When $f=h_L$ we put $a=\eta$ which matches the definition of $h_L$. Then,
\begin{equation*}
\begin{split}
  \frac{1}{2}h_L^*(L)
   & \leq L\Theta(L) \int_0^\eta s\, ds 
        + \sqrt{2}\vartheta(L)\int_0^{\frac{1}{\eta}} \frac{1}{1+s}\, ds\\
   & = \frac{1}{2} \frac{L\vartheta(L)^2}{\Theta(L)} 
        + \sqrt{2}\vartheta(L)\ln\big( 1+\frac{\Theta(L)}{\vartheta(L)} \big) 
\end{split}
\end{equation*}
where we employed the same substitution as above.
\end{proof}

We will need the delta-estimate only for $n=2$ and two choices of
$\vartheta$ and $\Theta$. The resulting estimates for the integral $I_L$ in \eqref{delta02t01} are
\begin{align}
  \vartheta(L)  =\ln L,\ \Theta(L)  =L: & &
     I_L & \leq C_1(W) \Big[ \frac{1}{L} + \frac{\ln L}{L^2} + \frac{\ln^2L}{L^2} \Big]\label{delta01}\\
  \vartheta(L)  =1,\     \Theta(L)  =L^{\frac{1}{2}}: & &
     I_L & \leq C_2(W) \Big[ \frac{1}{L^{\frac{3}{2}}} + \frac{\ln L}{L^{\frac{5}{2}}}
         + \frac{\ln L}{L^2}\Big]\label{delta02}
\end{align}
for $L\to\infty$ with constants $C_1(W)$, $C_2(W)$ that depend only on $W$.

%% file: asymptotics.tex
In the thermodynamic limit the particle density, $\rho$, is kept
constant. Usually, that would be $N/(2L)$. However, taking
\begin{equation}\label{thermodynamic_limit}
  \rho := \frac{N+\frac{1}{2}}{2L}>0,\ \nu:=\pi^2\rho^2
\end{equation}
will make our formulae handier since $\nu_N=\nu$ is constant then.
We start with combining Propositions \ref{anderson01t}, \ref{anderson02t},
and \ref{f_delta01t} and write
\begin{equation}\label{integral_formula}
   \tr \big[ P_N(\id-\Pi_M)\big] = \frac{1}{2\pi i}\int_{\Gamma_N} \Big[\frac{1}{2z} \tr A_N(z)\, dz
                  +  \frac{1}{z} \tr B_N(z)\Big] \, dz .
\end{equation}
Here $M=M(N)$ according to the decomposition
\eqref{p_decomposition_M} of the spectrum $\sigma(H_V)$. The Fermi parabola
$\Gamma_N$ is defined in \eqref{fermi_parabola}.
The operators in \eqref{integral_formula} are
\begin{equation}\label{asymptotics_AB}
  A_N(z) := P_NR(z)T(z) ( R(z) - C(z) ) T(z),\ 
  B_N(z) := P_NR(z) T(z) D(z) T(z),
\end{equation}
with the operators $P_N$, $R(z)$, $C(z)$, $D(z)$, and $T(z)$ defined in
\eqref{free_spectral_projection}, \eqref{f_resolvent_resolvent},
\eqref{anderson02t02}, \eqref{anderson_T-matrix}, respectively.
The traces can be treated further. For $A_N(z)$ we use the $\varphi_j$'s and write
\begin{equation}\label{asymptotics_A_N}
  \tr A_N(z)
    = \sum_{j=1}^N\frac{1}{z-\lambda_j}
              (\varphi_j, T(z)(R(z)-C(z))T(z)\varphi_j) .
\end{equation}
For $B_N(z)$ we recall the definition $S_N(z)=P_NR(z)$, the rank-one operators
$P_{s,c}(z)$ from \eqref{f_resolvent_P}, the operator $\Omega(z)$ from
\eqref{anderson_Omega}, and \eqref{f_delta01t02} for $D(z)$.
Then,
\begin{equation}\label{asymptotics_B_N}
\begin{split}
    \tr B_N(z)
         & =  d_{s,L}(z)(J\Omega(\bar z)\omega_s(\bar z),\sqrt{|V|}S_N(z)\sqrt{|V|}J\Omega(z)\omega_s(z))\\
         & \quad + d_{c,L}(z)(J\Omega(\bar z)\omega_c(\bar z),\sqrt{|V|}S_N(z)\sqrt{|V|}J\Omega(z)\omega_c(z)),
\end{split}
\end{equation}
with $\omega_{s,c}(z)$ as in \eqref{f_resolvent_omega} and the abbreviation
\begin{equation}\label{asymptotics_delta}
   d_{s,L}(z) := \frac{L}{4\sin^2(L\sqrt{z})},\
   d_{c,L}(z) := \frac{L}{4\cos^2(L\sqrt{z})} .
\end{equation}
The complex conjugates in \eqref{asymptotics_B_N} are due to the sesquilinearity
of the scalar product. We will see that both $\tr A_N(z)$ and $\tr B_N(z)$
decay sufficiently fast on the Fermi parabola such that the integrals can be
treated separately.

%% file: asymptotics_subdominant.tex
We discuss the subdominant term arising directly from the integral formula. 
Additional corrections will appear in Section \ref{asymptotics_d}.

\begin{proposition}\label{asymptotics_s01t}
Let $V\in L^1(\R)\cap L^\infty(\R)$ satisfy \eqref{p_resolvent01t00}. Furthermore, assume
$X^2V\in L^1(\R)$ with the operator $X$ from \eqref{nabla_X}. 
Then,
\begin{equation}\label{asymptotics_s01t01}
  \Big|\int_{\Gamma_N} \frac{1}{2z} \tr A_N(z)\, dz\Big| \leq 
    C_{sub} \frac{\sqrt{N+1}}{\sqrt{L}} 
               \big( \nu_N^{-\frac{5}{4}} + \nu_N^{-\frac{3}{4}}\big)
\end{equation}
with a constant $C_{sub}\geq 0$. The integral converges absolutely. The operator
$A_N(z)$ is defined in \eqref{asymptotics_AB} and the Fermi parabola $\Gamma_N$
in \eqref{fermi_parabola}.
\end{proposition}
\begin{proof}
We estimate $\tr A_N(z)$ for $z\in\Gamma_N$ (see \eqref{fermi_parabola}).
From \eqref{asymptotics_A_N} we obtain
\begin{equation}\label{asymptotics_s01t02}
  | \tr A_N(z) |
     \leq \sum_{j=1}^N\frac{1}{|z-\lambda_j|}\Big( 
        |(\varphi_j, T(z)R(z)T(z)\varphi_j)|
         + |(\varphi_j, T(z)C(z) T(z)\varphi_j)|\Big) .
\end{equation}
The matrix elements can be estimated with the aid of Lemmas \ref{f_resolvent02t}
and \ref{p_resolvent01t} as follows
\begin{multline*}
  |(\varphi_j, T(z)R(z)T(z)\varphi_j)|
       = |(\varphi_j,\sqrt{|V|}J\Omega(z)\sqrt{|V|}R(z)\sqrt{|V|}J\Omega(z)\sqrt{|V|}\varphi_j)| \\
       \leq \|\sqrt{|V|}\varphi_j\|^2 \|\Omega(z)\|^2 \|K(z)\|
       \leq \frac{C_1}{L}\frac{1}{\sqrt{\nu_N+s^2}} ,
\end{multline*}
and
\begin{equation*}
  |(\varphi_j, T(z)C(z)T(z)\varphi_j)| 
      \leq \|\sqrt{|V|}\varphi_j\|^2 \|\Omega(z)\|^2
       \|\sqrt{|V|}C(z)\sqrt{|V|}\|
      \leq \frac{C_2}{L}
\end{equation*}
with constants
\begin{equation*}
  C_1 :=  4\|V\|_1^2C_\Omega^2,\
  C_2 :=  8\|V\|_1^{\frac{3}{2}}\|X^2V\|_1^{\frac{1}{2}} C_\Omega^2 .
\end{equation*}
In order to treat the remaining sum in \eqref{asymptotics_s01t02} we bound
\eqref{fermi_parabola_distance} from below via
\begin{equation*}
   |z-\lambda_j|
     \geq (\nu_N+s^2)^{\frac{1}{2}}((\sqrt{\nu_N}-\sqrt{\lambda_j})^2+s^2)^{\frac{1}{2}}
     \geq \sqrt{2}(\nu_N+s^2)^{\frac{1}{2}}
         (\sqrt{\nu_N}-\sqrt{\lambda_j})^{\frac{1}{2}} \sqrt{|s|} .
\end{equation*}
With the aid of \eqref{estimates_sum} we obtain
\begin{equation*}
\begin{split}
   \sum_{j=1}^N \frac{1}{|z-\lambda_j|}
      & \leq \frac{1}{\sqrt{2}} \frac{1}{(\nu_N+s^2)^{\frac{1}{2}}}
           \frac{1}{\sqrt{|s|}} 
           \frac{\sqrt{2L}}{\sqrt{\pi}} \sum_{j=1}^N \frac{1}{(N+\frac{1}{2}-j)^{\frac{1}{2}}}\\
      & \leq \frac{4}{\sqrt{\pi}}
        \frac{\sqrt{L(N+1)}}{(\nu_N+s^2)^{\frac{1}{2}}\sqrt{|s|}}
\end{split}
\end{equation*}
for $s\neq 0$ and thus
\begin{equation}\label{asymptotics_s01t03}
  |\tr A_N(z)| \leq \frac{4}{\sqrt{\pi}} \frac{\sqrt{N+1}}{\sqrt{L}}
                           \frac{1}{(\nu_N+s^2)^{\frac{1}{2}}\sqrt{|s|}}
                     \Big( \frac{C_1}{(\nu_N+s^2)^{\frac{1}{2}}} +C_2\Big) .
\end{equation}
Parametrizing the Fermi parabola as usual (see \eqref{fermi_parabola}),
we estimate in \eqref{asymptotics_s01t01}
\begin{multline*}
    \Big| \int_\R \frac{1}{\sqrt{\nu_N}+is} \tr A_N(z(s))\, ds \Big| \\
      \leq \frac{4}{\sqrt{\pi}} \frac{\sqrt{N+1}}{\sqrt{L}}
            \Big[ \int_\R \frac{C_1}{(\nu_N+s^2)^{\frac{3}{2}}\sqrt{|s|}}\, ds
                 + \int_\R \frac{C_2}{(\nu_N+s^2)\sqrt{|s|}}\, ds
            \Big] ,
\end{multline*}
where we used \eqref{asymptotics_s01t03}. 
For $\alpha\in\{\frac{3}{2},1\}$ the integral
\begin{equation*}
  \int_\R \frac{1}{(\nu_N+s^2)^\alpha\sqrt{|s|}} \, ds
   = 4 \nu_N^{\frac{1}{4}-\alpha} \int_0^\infty \frac{1}{(1+s^4)^\alpha}\, ds
\end{equation*}
exists. Thus, the integral in \eqref{asymptotics_s01t01} converges absolutely
and satisfies the bound given there with an appropriate constant.
\end{proof}

%% file: asymptotics_dominant.tex
To begin with, we single out the dominant part of the integral over $\tr B_N(z)$.

\begin{proposition}\label{asymptotics_d01t}
Let $V\in L^1(\R)\cap L^\infty(\R)$ satisfy \eqref{p_resolvent01t00}, \eqref{p_resolvent02t00}, 
and $X^pV\in L^\infty(\R)$, $p=2,3$. The following integral over $B_N$ (see
\eqref{asymptotics_B_N}) converges absolutely and behaves in the thermodynamic
limit according to \eqref{thermodynamic_limit} asymptotically as
\begin{equation}\label{asymptotics_d01t01}
  \frac{1}{2\pi i}\int_{\Gamma_N} \frac{1}{z} \tr B_N(z)\, dz
     =  \varkappa_N \gamma_L(\nu) + O(1),\ N,L\to\infty .
\end{equation}
Here, $\varkappa_N$ is from Proposition \ref{f_truncated02t}, and
$\gamma_L(\nu):=\gamma_{s,L}(\nu)+\gamma_{c,L}(\nu)$,
\begin{equation}\label{asymptotics_d01t02}
  \gamma_{s,c,L}(\nu) :=
     \frac{1}{\pi\sqrt{\nu}} \int_\R d_{s,c,L}(s)
       ( \Omega_\infty(\nu,-Ls)\omega(\nu),
           F(\nu)\Omega_\infty(\nu,Ls)\omega(\nu))\, ds
\end{equation}
with the bounded operator
\begin{equation}\label{asymptotics_d01t01c}
  F(\nu) := J^* \sqrt{|V|}(P_s(\nu)+P_c(\nu))\sqrt{|V|}J .
\end{equation} 
$\Omega_\infty$ stands for $\Omega_\infty^+$, $\Omega_\infty^-$
depending on whether $N$ is even or odd.
\end{proposition}
\begin{proof}
We proceed in three steps. First, we show that the integral converges
absolutely. Then, we weed out the non-essential parts of the integral with the aid
of the delta-estimate, Proposition \ref{delta02t}.
Finally, we keep only the dominant part of the truncated resolvent $S_N(\nu)$.

(a) We bound $\tr B_N(z)$ for $z\in\Gamma_N$ (see \eqref{fermi_parabola}).
With the aid of Lemmas \ref{f_resolvent01t},
\ref{f_resolvent01bt}, and \ref{f_truncated01t} we infer from \eqref{asymptotics_B_N}
\begin{equation}\label{asymptotics_d01t03}
\begin{split}
  |\tr B_N(z)| 
     & \leq \|\Omega(\bar z)\| \|\Omega(z)\| \|\sqrt{|V|}S_N(z)\sqrt{|V|}\| \times\\
     & \quad \times (|d_s(z)| \|\omega_s(\bar z)\| \|\omega_s(z)\| + 
               |d_c(z)| \|\omega_c(\bar z)\| \|\omega_c(z)\|) \\
     & \leq \frac{32}{\pi} C_\Omega^2 \|V\|_1 \frac{e^{-2L|s|}}{\sqrt{\nu+s^2}}
             V_L(2s)\ln(N+1) .
\end{split}
\end{equation}
Parametrizing the Fermi parabola as in \eqref{fermi_parabola}, we conclude that
\begin{equation}\label{asymptotics_d01t01a}
  \frac{1}{2\pi i}\int_{\Gamma_N} \frac{1}{z} \tr B_N(z)\, dz
     = \frac{1}{\pi} \int_\R \frac{1}{\sqrt{\nu}+is} \tr B_N(z(s))\, ds
\end{equation}
converges absolutely because of \eqref{delta01t01} with $g\equiv 1$ and $m=n=0$.

(b) The following calculations look alike for 
$\delta_{s,c,L}(s) := d_{s,c}(z(s))$ 
and the corresponding quantities. Therefore, we simply write
$\delta_L(s)$ etc. to denote either case. We evaluate the integral in
\eqref{asymptotics_d01t01a} by successively simplifying
\begin{equation}\label{asymptotics_d01t01b}
  (J\Omega(\bar z(s))\omega(\bar z(s)),
    \sqrt{|V|}S_N(z(s))\sqrt{|V|}J\Omega(z(s))\omega(z(s)))
\end{equation}
in the integrand with the aid of the delta-estimate, Proposition
\ref{delta02t}.

(i) At first we replace $S_N(z)$  in \eqref{asymptotics_d01t01b} by $S_N(\nu)$
which results in the error
\begin{equation*}
\begin{split}
  e^{(1)}_L
  & :=
  \Big|
     \int_\R \frac{\delta_L(s)}{\sqrt{\nu}+is} \times \\
  & \quad (J\Omega(\bar z(s))\omega(\bar z(s)),\sqrt{|V|}(S_N(z(s))-S_N(\nu))\sqrt{|V|}J\Omega(z(s))\omega(z(s)))\,
          ds\Big| \\
  & \leq 4C_\Omega^2 L \int_\R \frac{e^{-2L|s|}}{\sqrt{\nu+s^2}}
          V_L(2s)
            \|\sqrt{|V|}(S_N(z(s))-S_N(\nu))\sqrt{|V|}\|\, ds .
\end{split}
\end{equation*}
Note that $N=N(L)$. Recalling Lemma \ref{f_truncated01t} we use Proposition \ref{delta02t} with
\begin{equation*}
  f_L(s) = \|\sqrt{|V|}(S_N(z(s))-S_N(\nu))\sqrt{|V|}\|,\
   \vartheta(L) = \ln L,\ \Theta(L) = L
\end{equation*}
and obtain the error (cf. \eqref{delta01})
\begin{equation*}
  e^{(1)}_L \leq C_1 \big( 1 + \frac{\ln L}{L} + \frac{\ln^2 L}{L}\big) .
\end{equation*} 
(ii) Now we replace the right $\omega(z)$ in \eqref{asymptotics_d01t01b} by $\omega(\nu)$ 
resulting in the error
\begin{equation*}
  e^{(2,r)}_L
   \leq 4C_\Omega^2 L \int_\R \frac{e^{-2L|s|}}{\sqrt{\nu+s^2}}
            \|\omega(\bar z(s))\| \|\omega(z(s))-\omega(\nu)\| \, ds \,
               \|\sqrt{|V|}S_N(\nu)\sqrt{|V|}\| .
\end{equation*}
By virtue of \eqref{f_resolvent01bt02} we can estimate
\begin{equation*}
  \|\omega(\bar z(s))\| \|\omega(z(s))-\omega(\nu)\|
      \leq (V_L(2|s|)^{\frac{1}{2}} + V_L(0)^{\frac{1}{2}} )
               V_L(2|s|)^{\frac{1}{2}}
      \leq 2 V_L(2|s|) .
\end{equation*}
Alternatively, \eqref{f_resolvent01bt03} along with \eqref{V_L_estimates} yields
\begin{equation*}
\begin{split}
  \|\omega(\bar z(s))\| \|\omega(z(s))-\omega(\nu)\|
    & \leq |s| V_L(2|s|)^{\frac{1}{2}} V_L^{(2)}(2|s|)^{\frac{1}{2}} \\
    & \leq |s| L^{\frac{1}{2}} V_L(2|s|)^{\frac{1}{2}}
           V_L^{(1)}(2|s|)^{\frac{1}{2}} .
\end{split}
\end{equation*}
Define $W$ by $W(x):=\max\{ |V(x)|, |xV(x)|\}$ and note $X^2W\in L^\infty(\R)$. Then, 
\begin{equation*}
  \|\omega(\bar z(s))\| \|\omega(z(s))-\omega(\nu)\|
    \leq f_L(s) W_L(2|s|) .
\end{equation*}
Thus, Proposition \ref{delta02} applies with $\vartheta(L)=1$,
$\Theta(L)=L^{\frac{1}{2}}$. The left $\omega(z(s))$ in
\eqref{asymptotics_d01t01b}
and the corresponding error $e_L^{(2,l)}$ can be treated in like
manner when one uses, for the sake of convenience, the same bound for $\omega(\nu)$
as for $\omega(z(s))$ (see \eqref{V_L_estimates}). Thus, the total error
$e_L^{(2)}:= e_L^{(2,r)}+e_L^{(2,l)}$ made in this section can be bounded
\begin{equation*}
  e^{(2)}_L \leq C_2 \Big( \frac{1}{L^{\frac{3}{2}}} 
                       + \frac{\ln L}{L^{\frac{5}{2}}}
                       + \frac{\ln L}{L^2} \Big) \ln L
\end{equation*}
where the rightmost logarithm is from the truncated resolvent (Lemma \ref{f_truncated01t}).

(iii) Finally, we replace $\Omega(z(s))$ in
\eqref{asymptotics_d01t01b} by $\Omega_\infty(\nu,Ls)$. 
Here, $\Omega_\infty$ stands for $\Omega_\infty^+$,
$\Omega_\infty^-$ when $N$ is even or odd, respectively. The inequality
\begin{equation*}
  \| \Omega(z(s)) - \Omega_\infty(\nu,Ls)\| \leq f_L(s) W_L(2|s|)
\end{equation*}
follows from \eqref{p_resolvent02t01} and \eqref{p_resolvent02t02}
with the same $W$ and the same simplifications as in (ii). The functions $f_L$ satisfy the
assumptions of Proposition
\ref{delta02t} with $\vartheta(L)=1$ and $\Theta(L)=L^\frac{1}{2}$. Hence,
$e^{(3)}_L$ can be bounded as in (ii).

(iv) It is easy to replace $\sqrt{\nu+s^2}$ in the integral
\eqref{asymptotics_d01t01a} by $\sqrt{\nu}$ which gives the error
\begin{equation*}
  e^{(4)}_L \leq 4 C_{\Omega_\infty}^2 \|\omega(\nu)\|^2 L
                     \int_\R \frac{e^{-2L|s|}|s|}{\sqrt{\nu(\nu+s^2)}}\,
                     ds \times \|\sqrt{|V|}S_N(\nu)\sqrt{|V|}\|
            \leq C_4 \frac{\ln L}{L} .
\end{equation*}
(c)
We decompose $S_N(\nu)$ according to Proposition \ref{f_truncated02t} and find
\begin{equation*}
  \|\sqrt{|V|}\tilde S_{j,N}(\nu)\sqrt{|V|}\|\leq C,\ j=0,1,\
  \|\sqrt{|V|} S_{1,N}(\nu)\sqrt{|V|}\|\leq C
\end{equation*}
because of the estimates \eqref{f_truncated01} and \eqref{f_truncated02}.
Thus, we are left with
\begin{equation*}
  \Big| \int_\R \delta_L(s) \big(J\Omega_\infty(\nu,-Ls)\omega(\nu),
           J\Omega_\infty(\nu,Ls)\omega(\nu)\big) \, ds \Big|
     \leq 4 C_{\Omega_\infty}^2 \|V\|_1 .
\end{equation*}
Hence, the dominant term is given through $S_{0,N}(\nu)$. Writing it as in
\eqref{f_truncated03} gives the operator $F(\nu)$, which is obviously bounded, and
thus $\gamma_L(\nu) = \gamma_{s,L}(\nu)+\gamma_{c,L}(\nu)$ with
$\gamma_{s,c,L}(\nu)$ as in \eqref{asymptotics_d01t02}.
Summing up the errors made in (i) through (iv) and in (c)
gives the overall error
\begin{equation*}
   |e_L| \leq C \Big( 1 + \frac{\ln L}{L} + \frac{\ln^2 L}{L}
                       +  \frac{\ln L}{L^{\frac{3}{2}}} 
                       + \frac{\ln^2 L}{L^{\frac{5}{2}}}
                       + \frac{\ln^2 L}{L^2} \Big)
\end{equation*}
which proves \eqref{asymptotics_d01t01}.
\end{proof}

The coefficient $\gamma_L(\nu)$ in \eqref{asymptotics_d01t01} seems to depend still
on $L$. We will see that this is actually not so.

\begin{theorem}\label{asymptotics_d02t}
Let $V\in L^1(\R)\cap L^\infty(\R)$ satisfy the assumptions of Propositions
\ref{asymptotics_s01t}, \ref{asymptotics_d01t} and in addition
\eqref{p_resolvent03t01} as well as
\eqref{p_eigenvalues01t01} with some $\alpha>0$. In the thermodynamic limit
according to \eqref{thermodynamic_limit} with Fermi energy $\nu$, the Anderson
integral \eqref{anderson_inequality} has the leading asymptotics
\begin{equation}\label{asymptotics_d02t01}
  \mathcal{I}_{N,L} = \gamma(\nu) \ln N + O(1),\ N,L\to\infty,
\end{equation}
with the constant
\begin{equation}\label{asymptotics_d02t02}
  \gamma(\nu) := \frac{1}{4\pi^2\nu}
        \tr\big[ 
        (\id+\frac{1}{4\nu}\hat\Phi(\nu)^2)^{-1}\hat\Phi(\nu)^2
           \big] \geq 0
\end{equation}
and the $2\times 2$ matrix $\hat\Phi(\nu)$ as in \eqref{omega_hat2}.
\end{theorem}
\begin{proof}
(a) 
We evaluate the integral in \eqref{asymptotics_d01t02}. First of all note the
operators $\Omega_\infty^\pm$. It will turn out that both $\Omega_\infty^+$ and
$\Omega_\infty^-$ eventually yield the same $\gamma(\nu)$. Therefore, we
restrict ourselves to $\Omega_\infty^+$ and drop the superscript for the sake of
convenience. 
We recall the definition \eqref{asymptotics_delta} of $d_{s,c,L}$ along with
\eqref{sine_cosine_fermi_parabola} and make a change of variables,
$s=t/L$. Then,
\begin{equation*}
\begin{split}
  \lefteqn{\gamma_{s,c}(\nu)}\\ 
   & = \frac{1}{8\pi^2\nu}\int_\R \frac{2}{(\cosh t\pm i\sinh t)^2} 
               \big(\Omega_\infty(\nu,-t)\omega_{s,c}(\nu),F(\nu)\Omega_\infty(\nu,t)\omega_{s,c}(\nu)\big)\,
               dt
\end{split}
\end{equation*}
where we dropped the index $L$ since there is no explicit $L$-dependence any longer.
With the definition \eqref{asymptotics_d01t01c} of $F(\nu)$ we obtain
\begin{equation*}
\begin{split}
  \lefteqn{(\Omega_\infty(\nu,-t)\omega_{s,c}(\nu),F(\nu)\Omega_\infty(\nu,t)\omega_{s,c}(\nu))}\\
   &  = (J\Omega_\infty(\nu,-t)\omega_{s,c}(\nu),\sqrt{|V|}P_s(\nu)\sqrt{|V|}J\Omega_\infty(\nu,t)\omega_{s,c}(\nu))\\
   & \quad + (J\Omega_\infty(\nu,-t)\omega_{s,c}(\nu),\sqrt{|V|}P_c(\nu)\sqrt{|V|}J\Omega_\infty(\nu,t)\omega_{s,c}(\nu))\\
   &  =
     (J\Omega_\infty(\nu,-t)\omega_{s,c}(\nu),\omega_s(\nu))(\omega_s(\nu),J\Omega_\infty(\nu,t)\omega_{s,c}(\nu))\\
   & \quad +
     (J\Omega_\infty(\nu,-t)\omega_{s,c}(\nu),\omega_c(\nu))(\omega_c(\nu),J\Omega_\infty(\nu,t)\omega_{s,c}(\nu)) .
\end{split}
\end{equation*}
With the aid of the matrices $\hat\Omega_\infty(\nu,t)$ and $\hat\tau(t)$
(cf. \eqref{omega_hat1} and \eqref{p_resolvent03t03}) we can write the
integrand of $\gamma_s(\nu)+\gamma_c(\nu)$ as the trace of $2\times 2$ matrices which
leads to the integral
\begin{equation*}
\begin{split}
  I   & := i \int_\R \tr\big[
          \hat\tau'(t)\hat\Omega_\infty(\nu,-t)^*\hat\Omega_\infty(\nu,t)\big]\, dt \\
      & = i \int_\R \tr \big[ \hat\tau'(t)
            \hat Z(\nu,t)\hat\Phi(\nu)\hat Z(\nu,t)\hat\Phi(\nu)\big]
            \,dt .
\end{split}
\end{equation*}
Here we used both equalities in \eqref{p_resolvent03t04} to express $\hat\Omega_\infty(\nu,t)$ through
$\hat\Phi(\nu)$ and $\hat Z(\nu,t)$. By the cyclicity of the trace,
\begin{equation*}
\begin{split}
  I & = i \int_\R \tr\big[
         \hat Z(\nu,t)\hat\Phi(\nu)\hat\tau'(t)\hat Z(\nu,t)\hat\Phi(\nu)\big]\, dt\\
    & = 2\sqrt{\nu} i \int_\R \tr\big[\hat Z'(\nu,t)\hat\Phi(\nu)\big] \, dt \\
    & = 2\sqrt{\nu} i \lim_{t\to\infty} \tr\big[ (\hat Z(\nu,t) - \hat Z(\nu,-t))\hat\Phi(\nu)\big] .
\end{split}
\end{equation*}
We compute the difference
\begin{equation*}
\begin{split}
  \hat Z(\nu,t) - \hat Z(\nu,-t)
    & = \frac{1}{2\sqrt{\nu}} \hat Z(\nu,t)\hat\Phi(\nu)(\hat\tau(t)-\hat\tau(-t))\hat Z(\nu,-t)\\
    & = \frac{2i\im\tau(t)}{2\sqrt{\nu}} \hat Z(\nu,t)\hat\Phi(\nu)
        \hat Z(\nu,-t) .
\end{split}
\end{equation*}
Thus, our integral becomes
\begin{equation*}
\begin{split}
  I & = - 2 \lim_{t\to\infty} \im(\tau(t))
         \tr\big[\hat Z(\nu,t)\hat\Phi(\nu)\hat Z(\nu,-t)\hat\Phi(\nu)\big]\\
    & = - 2 \lim_{t\to\infty} \im(\tau(t))\tr\big[\hat Z(\nu,t)\hat\Phi(\nu)^2\hat Z(\nu,t)^*\big],
\end{split}
\end{equation*}
where we used \eqref{p_resolvent03t04}. 
The limit can be computed via \eqref{f_resolvent_tau} and \eqref{p_resolvent03t03}. Then,
\begin{equation*}
\begin{split}
  I & = 2 \tr\big[ (\id-\frac{i}{2\sqrt{\nu}}\hat\Phi(\nu))^{-1}
               (\id+\frac{i}{2\sqrt{\nu}}\hat\Phi(\nu))^{-1}
               \hat\Phi(\nu)^2\big]\\
    & = 2 \tr\big[ (\id+\frac{1}{4\nu}\hat\Phi(\nu)^2)^{-1}\hat\Phi(\nu)^2\big] .
\end{split}
\end{equation*}
Apart from the prefactor this is the coefficient $\gamma(\nu)$ in \eqref{asymptotics_d02t02}.
From $\hat\Phi(\nu)^*=\hat\Phi(\nu)$ we infer that $\gamma(\nu)$ is the trace
of the product of two non-negative matrices. Hence $\gamma(\nu)\geq 0$.

(b)
The integral formula \eqref{integral_formula} involves $\Pi_M$ instead of
$\Pi_N$ and therefore differs from the actual Anderson integral
$\mathcal{I}_{N,L}$ by
\begin{equation*}
  | \tr P_N(\id-\Pi_N) - \tr P_N(\id-\Pi_M) |
    = | \tr P_N(\Pi_N-\Pi_M) |
    \leq |N-M| .
\end{equation*}
From Propositions \ref{p_eigenvalues01t} and \ref{p_eigenvalues02t} we deduce
\begin{equation*}
   N-\frac{1}{2} - \frac{1}{\nu} \frac{2}{\pi}\|V_+\|_1
       \leq M \leq N + \frac{1}{2} + 
                  \frac{1}{2\nu}\Big( 
          \frac{2C_\alpha}{\alpha\pi} (\|V_-\|_\infty+\nu)^{\frac{1}{2}} + \|V_-\|_\infty
          \Big) .
\end{equation*}
Therefore, replacing $\Pi_N$ by $\Pi_M$ causes an error that is bounded by a constant.
Now, Propositions \ref{asymptotics_s01t} and \ref{asymptotics_d01t} along with
the asymptotics for $\varkappa_N$ in Lemma \ref{estimates01t} prove
\eqref{asymptotics_d02t01}.
\end{proof}

The coefficient $\gamma(\nu)$ can be given a scattering theoretical
interpretation. Recall that in this one-dimensional case the S-matrix is 
indeed a $2\times 2$-matrix,
\begin{equation}\label{scattering_matrix}
  S(\nu) =
\begin{pmatrix}
   t(\sqrt{\nu})   & r_2(\sqrt{\nu}) \\
   r_1(\sqrt{\nu}) & t(\sqrt{\nu})
\end{pmatrix}
\end{equation}
with the transmission coefficient $t(\sqrt{\nu})$ and the reflection
coefficients $r_{1,2}(\sqrt{\nu})$ (e.g. \cite{DeiftTrubowitz1979}, in
particular pp. 143--146 for the formulae needed herein).
In what follows we drop the $\nu$ in the argument of
operators and vectors which makes the formulae look a little less ornate. To begin
with, we decompose $K$ into a Lippmann-Schwinger like operator and a rank two operator
\begin{equation*}
  \sqrt{|V|}K\sqrt{|V|} 
    = - \sqrt{|V|}K_+\sqrt{|V|} + \frac{1}{2\sqrt{\nu}}(\omega_c,\cdot)\omega_s
                                - \frac{1}{2\sqrt{\nu}}(\omega_s,\cdot)\omega_c
\end{equation*}
by using the addition theorem for the sine. The operator $K_+$ has the kernel
\begin{equation*}
  K_+(x,y) := \frac{1}{\sqrt{\nu}} \chi(y-x)\sin(\sqrt{\nu}(x-y)),\ x,y\in\R,
\end{equation*}
with the Heaviside function $\chi$ being zero for $x<0$ and one
elsewhere. We define further
\begin{equation*}
  \Phi_+ := (\id + \sqrt{|V|}K_+\sqrt{|V|}J)^{-1},\
  \hat\Phi_+ :=
\begin{pmatrix}
  (\omega_s,J\Phi_+\omega_s) & (\omega_s,J\Phi_+\omega_c) \\
  (\omega_c,J\Phi_+\omega_s) & (\omega_c,J\Phi_+\omega_c)
\end{pmatrix} .
\end{equation*}
We will see below that the entries of $\hat\Phi_+$ can be computed explicitly
with the aid of the transmission and reflection coefficients. We want to express
$\Phi$ through $\Phi_+$, which amounts to solving the equation
\begin{equation*}
  (\id+\sqrt{|V|}K_+\sqrt{|V|}J)\psi
       - \frac{1}{2\sqrt{\nu}}(\omega_c,J\psi)\omega_s
       + \frac{1}{2\sqrt{\nu}}(\omega_s,J\psi)\omega_c
    = \omega
\end{equation*}
for $\psi$. Here, $\omega$ equals $\omega_s$ or $\omega_c$. Since we are only interested in
$\hat\Phi$ we take scalar products and obtain after some elementary calculations
\begin{equation*}
  ( \id + \frac{1}{2\sqrt{\nu}} \hat\Phi_+ W)\hat\Phi = \hat\Phi_+\
  \text{with}\
  W :=
\begin{pmatrix}
  0 & -1 \\
  1 & 0
\end{pmatrix}
  ,\ W^2 = -\id.
\end{equation*}
We assume the first matrix to be invertible,
\begin{equation*}
  \hat\Phi = (\id + \frac{1}{2\sqrt{\nu}}\hat\Phi_+W)^{-1}\hat\Phi_+
           = \hat\Phi_+(\id + \frac{1}{2\sqrt{\nu}}W\hat\Phi_+)^{-1},
\end{equation*}
and obtain
\begin{equation*}
  4\pi^2\nu\gamma = \tr\big[ (\id+\frac{1}{4\nu}\hat\Phi^2)^{-1}\hat\Phi^2\big]
     = \tr\big[ \hat\Phi_+(\id + \frac{1}{2\sqrt{\nu}}\hat\Phi_+W 
           + \frac{1}{2\sqrt{\nu}}W\hat\Phi_+)^{-1}\hat\Phi_+\big] .
\end{equation*}
Scattering theory in general uses exponential functions,
\begin{equation*}
    e_\pm(x) := \sqrt{|V(x)|}\, e^{\pm i\sqrt{\nu}x},
\end{equation*}
rather than the trigonometric functions as in $\omega_{s,c}$. Thus, we introduce
\begin{equation*}
  \tilde\Phi_+ :=
\begin{pmatrix}
  (e_+, J\Phi_+ e_+) & (e_+, J\Phi_+ e_-) \\
  (e_-, J\Phi_+ e_+) & (e_-, J\Phi_+ e_-)
\end{pmatrix}  
  = 2i\sqrt{\nu}
\begin{pmatrix}
  1 - \frac{1}{t} & -\frac{\bar r_2}{\bar t} \\
  \frac{r_2}{t}   & \frac{1}{\bar t} -1 
\end{pmatrix}
\end{equation*}
with $r_{1,2}$ and $t$ from \eqref{scattering_matrix}
(see \cite{DeiftTrubowitz1979}, pp. 145, 146). We transform our matrices
\begin{equation*}
  \hat\Phi_+ =\frac{1}{2} U^*\tilde\Phi_+ U,\
  W = iU^*IU \ \text{with}\
  U:= \frac{1}{\sqrt{2}}
\begin{pmatrix}
  -i & 1 \\
   i & 1
\end{pmatrix}
  ,\
  I :=
\begin{pmatrix}
  1 & 0\\
  0 & -1
\end{pmatrix}
\end{equation*}
and the trace becomes
\begin{equation*}
  4\pi^2\nu\gamma
    = \frac{1}{4} \tr\big[ (\id+\frac{i}{4\sqrt{\nu}}I\tilde\Phi_+
       + \frac{i}{4\sqrt{\nu}}\tilde\Phi_+ I)^{-1}\tilde\Phi_+^2\big] .
\end{equation*}
The inverse simplifies considerably since $I\tilde\Phi_+ + \tilde\Phi_+ I$ is
diagonal. Thereby,
\begin{equation*}
   (\id+\frac{i}{4\sqrt{\nu}}I\tilde\Phi_+ + \frac{i}{4\sqrt{\nu}}\tilde\Phi_+ I)^{-1}
     = 
\begin{pmatrix}
   t & 0 \\
   0 & \bar t
\end{pmatrix} .
\end{equation*}
Furthermore,
\begin{equation*}
  \tilde\Phi_+^2 = 
    -4\nu
\begin{pmatrix}
  ( 1-\frac{1}{t})^2 - \big|\frac{r_2}{t}\big|^2 & * \\
     * & ( \frac{1}{\bar t}-1)^2 - \big|\frac{r_2}{t}\big|^2)
\end{pmatrix}
\end{equation*}
where the off-diagonal elements are not needed. Finally,
\begin{equation*}
  4\pi^2\nu\gamma
   = -2\nu \re \Big\{ t \big[ (1-\frac{1}{t})^2 -
          \big|\frac{r_2}{t}\big|^2\big] \Big\}
   = 4\nu \re(1-t)
\end{equation*}
where we used $|t|^2+|r_2|^2=1$ which is due to the unitarity of the
S-matrix. We summarize what we have found.

\begin{corollary}\label{asymptotics_d03t}
The coefficient $\gamma(\nu)$ in Theorem \ref{asymptotics_d02t} can be written
\begin{equation*}
  \gamma(\nu) = \frac{1}{\pi^2}(1- \re t(\sqrt{\nu}))
\end{equation*}
where $t(\sqrt{\nu})$ is the transmission coefficient with wave number $\sqrt{\nu}$.
\end{corollary}

In \cite{GebertKuettlerMueller2013}, Theorem 2.4, the lower bound
\begin{equation*}
  \mathcal{I}_{N,L}\geq \gamma'(\nu)\ln N,\ 
  \gamma'(\nu)=\frac{1}{(2\pi)^2}\tr \big[ (S(\nu)-\id)^*(S(\nu)-\id) \big]
\end{equation*}
has been derived where $S(\nu)$ is the S-matrix at energy
$\nu$. By Corollary \ref{asymptotics_d03t}, $\gamma'(\nu)=\gamma(\nu)$ in
one-dimension.

%% file: determinant.tex
The asymptotics in Theorem \ref{asymptotics_d02t} can be used to derive lower
and upper bounds for the transition probability $\mathcal{D}_N$ from
\eqref{transition_probability}. Standard reasoning yields
\begin{equation}\label{det01}
  \mathcal{D}_N = \det P_N\Pi_N P_N = \exp(\tr\ln (P_N\Pi_N P_N))
\end{equation}
where the determinant is to be taken with respect to $\ran P_N$
otherwise it would be zero. Using Wouk's integral formula \cite{Wouk1965} for
the operator logarithm (see also \cite{Nollau1969}) we obtain
\begin{equation}\label{det02}
  \mathcal{D}_N
   = \exp\Big[-\int_0^1 \tr\big[ P_N(\id-P_N\Pi_N P_N)(\id-t(\id-P_N\Pi_N P_N))^{-1}\big]\, dt\Big] ,
\end{equation}
which immediately yields the inequalities
\begin{multline}\label{det03}
  \exp\big[- (1-\|P_N(P_N-\Pi_N)P_N\|)^{-1} \tr P_N(\id-\Pi_N)\big] \\
     \leq   \mathcal{D}_N
       \leq \exp\big[ -\tr P_N(\id-\Pi_N) \big] .
\end{multline}
The upper bound was already derived by Anderson \cite{Anderson1967} using
Hadamard's and Bessel's inequality as well as an inequality for the
logarithm. The lower bound, of course, holds only true when $\|P_N(P_N-\Pi_N)P_N\|<1$. 
Such operator-norm estimates are studied in the realm of so-called subspace
perturbation problems. However, those
results either depend on the size of the spectral gap (see \cite{McEachin1993})
or require perturbations that are off-diagonal with respect to $P_N$ (see
\cite{MotovilovSelin2006}). Both conditions are not met here wherefore we 
present a new approach.

\begin{theorem}\label{det01t}
Let $V\in L^1(\R)\cap L^\infty(\R)$ satisfy \eqref{p_resolvent01t00}.
Moreover, assume that the assumptions of Propositions \ref{p_eigenvalues01t} and
\ref{p_eigenvalues02t} are satisfied such that
\begin{equation}\label{det01t01}
  \frac{1}{2\nu_N}\Big(
    \frac{2C_\alpha}{\alpha\pi}(\|V_-\|_\infty+E)^{\frac{1}{2}} + \|V_-\|_\infty
      \Big) < \frac{1}{2}
  \ \text{and} \
  \frac{1}{\sqrt{\nu_N}} \frac{2\|V_+\|_1}{\pi} < \frac{1}{2}
\end{equation}
with some $\alpha>0$. Then,
\begin{equation}\label{det01t02}
  \|P_N(P_N-\Pi_N)P_N\| \leq \frac{16 C_\Omega}{\sqrt{\nu_N}} \|V\|_1 .
\end{equation}
\end{theorem}
\begin{proof}
Because of \eqref{det01t01} and Propositions \ref{p_eigenvalues01t} and
\ref{p_eigenvalues02t} we may compute the matrix elements
$a_{jk}:=(\varphi_j,(P_N-\Pi_N)\varphi_k)$ via the integral formula
\eqref{anderson01t02a}
\begin{equation*}
   a_{jk}
     = \frac{1}{\pi} \int_\R
        \frac{\sqrt{\nu_N}+is}{(z(s)-\lambda_j)(z(s)-\lambda_k)}
        (\varphi_j,\sqrt{|V|}J\Omega(z(s))\sqrt{|V|}\varphi_k)\, ds
\end{equation*}
where $z(s)\in\Gamma_N$ (see \eqref{fermi_parabola}).
By \eqref{free_eigenfunction_delocalized} and \eqref{p_resolvent01t01}
these can be estimated
\begin{equation*}
\begin{split}
  |a_{jk}|
     & \leq C_L
          \int_0^\infty \frac{1}{\sqrt{1+s^2}} 
                  \frac{1}{( (1-\frac{j}{N+\frac{1}{2}})^2+s^2 )^{\frac{1}{2}}}
                  \frac{1}{( (1-\frac{k}{N+\frac{1}{2}})^2+s^2 )^{\frac{1}{2}}}
                   \, ds \\
     & =: C_L b_{jk},\ C_L := \frac{2\|V\|_1 C_\Omega}{\pi\nu_N}\frac{1}{L}.
\end{split}
\end{equation*}
By the variational principle
\begin{equation*}
  \|A\| \leq C_L \|B\|,\ 
    A:=(a_{jk})_{j,k=1,\ldots,N},\ B:=(b_{jk})_{j,k=1,\ldots,N}.
\end{equation*}
We introduce the integral operator $k_N:L^2(\R^+)\to L^2(\R^+)$ with kernel
\begin{equation*}
  k_N(s,t) :=
    \frac{1}{(1+s^2)^{\frac{1}{4}}(1+t^2)^{\frac{1}{4}}}
      \sum_{j=1}^N  \frac{1}{( (1-\frac{j}{N+\frac{1}{2}})^2+s^2 )^{\frac{1}{2}}}
                    \frac{1}{( (1-\frac{j}{N+\frac{1}{2}})^2+t^2 )^{\frac{1}{2}}}.
\end{equation*}
Simple algebra shows that each eigenvalue of $B$ is an eigenvalue of $k_N$ as
well. It is therefore enough to bound the operator norm $\|k_N\|$. To this end, we
drop the prefactor and estimate the sum by an integral
\begin{equation*}
\begin{split}
   k_N(s,t)
    & \leq 
       2
       \int_0^{N+\frac{1}{2}} \frac{1}{( (1-\frac{u}{N+\frac{1}{2}})^2+s^2 )^{\frac{1}{2}}}
                    \frac{1}{( (1-\frac{u}{N+\frac{1}{2}})^2+t^2 )^{\frac{1}{2}}}
                      \, du\\
    & = 2(N+\frac{1}{2}) \int_0^1 \frac{1}{(u^2+s^2)^{\frac{1}{2}}(u^2+t^2)^{\frac{1}{2}}}
                      \, du\\
    & =: 2(N+\frac{1}{2}) k(s,t) .
\end{split}
\end{equation*}
Once again, it is enough to bound $\|k\|$. To this end, we estimate the quadratic
form of $k$ by using Cauchy's inequality along with Hilbert's trick 
\begin{equation*}
\begin{split}
  |(f,kf)|
    & = \int_0^\infty\int_0^\infty
          \Big(\frac{s}{t}\Big)^{\frac{1}{4}}f(s)\sqrt{k(s,t)}
          \Big(\frac{t}{s}\Big)^{\frac{1}{4}}f(t)\sqrt{k(s,t)}\, dt\, ds \\
    & \leq \int_0^\infty f(s)^2 \int_0^\infty\Big(\frac{s}{t}\Big)^{\frac{1}{2}}
          k(s,t)\, dt\, ds .
\end{split}
\end{equation*}
Note that $k(s,t)=k(t,s)$. We evaluate the $t$-integral
\begin{equation*}
\begin{split}
  \int_0^\infty \Big( \frac{s}{t} \Big)^{\frac{1}{2}} k(s,t)\, dt
     & = s^{\frac{1}{2}} \int_0^1 \frac{1}{(u^2+s^2)^{\frac{1}{2}}}
         \int_0^\infty
         \frac{1}{(u^2+t^2)^{\frac{1}{2}}}\frac{1}{t^{\frac{1}{2}}}\, dt\, du\\
     & = \int_0^{\frac{1}{s}}
         \frac{1}{(1+u^2)^{\frac{1}{2}}}\frac{1}{u^{\frac{1}{2}}}\, du
                  \int_0^\infty
         \frac{1}{(1+t^2)^{\frac{1}{2}}}\frac{1}{t^{\frac{1}{2}}}\, dt\\
     & \leq 4 \Big[ \int_0^\infty \frac{1}{(1+t^4)^{\frac{1}{2}}}\, dt\Big]^2 .
\end{split}
\end{equation*}
The last integral could be expressed with the aid of the gamma
function. However, since a bound is enough we estimate the integrand by means of $1+t^2$
to obtain $\|k\| \leq 2\pi^2$. This concludes the proof. 
\end{proof}

\begin{corollary}\label{det02t}
Let the conditions of Theorems \ref{asymptotics_d02t} and \ref{det01t}
be satisfied. Assume further that $\|V\|_1$ and $\nu$ are such that
$\|P_N(P_N-\Pi_N)P_N\|<1$ in \eqref{det01t02}. Then, the transition probability $\mathcal{D}_{N,L}$
(cf. \eqref{transition_probability}) satisfies in the thermodynamic limit
(cf. \eqref{thermodynamic_limit})
\begin{equation*}
  \tilde C N^{-\tilde\gamma(\nu)} \leq \mathcal{D}_{N,L} \leq C N^{-\gamma(\nu)}
\end{equation*}
with appropriate constants $\tilde C,C >0$, $\gamma(\nu)$ from Theorem
\ref{asymptotics_d02t}, and $\tilde\gamma(\nu)>0$.
\end{corollary}
\begin{proof}
The upper bound follows from \eqref{det03} and Theorem
\ref{asymptotics_d02t}. For the lower bound one needs in addition Theorem
\ref{det01t} which also gives $\tilde\gamma>0$.
\end{proof}

%% file: appendix_estimates.tex
At various points we need estimates which are not directly related to our main
subject. To begin with, we mention the following sums
\begin{equation*}
  \sum_{j=1}^N \frac{1}{(N+\frac{1}{2}-j)^\alpha}
    = \sum_{j=0}^{N-1} \frac{2^\alpha}{(2j + 1 )^\alpha} 
    \leq 4^\alpha \sum_{j=1}^N \frac{1}{(j+1)^\alpha}
    \leq 4^\alpha \int_0^N \frac{1}{(t+1)^\alpha} \, dt .
\end{equation*}
Evaluating the integral yields
\begin{equation}\label{estimates_sum}
   \sum_{j=1}^N \frac{1}{(N+\frac{1}{2}-j)^\alpha} \leq 4^\alpha
\begin{cases}
   \frac{1}{\alpha-1}                 & \text{for}\ \alpha>1 , \\
   \frac{1}{1-\alpha}(N+1)^{1-\alpha} & \text{for}\ 0\leq\alpha<1 , \\
   \ln(N+1)                           & \text{for}\ \alpha=1 .
\end{cases}
\end{equation}
The constant $\varkappa_N$ in Proposition \ref{f_truncated02t} requires more reasoning.

\begin{lemma}\label{estimates01t}
Let $M\in\R$, $w\in\C$ such that $\re w+\frac{1}{2}-M>0$. Then,
\begin{equation*}
  \int_0^\infty e^{-wt}\frac{\cosh (Mt)}{\cosh\frac{t}{2}}\, dt
    \leq \frac{2}{w+\frac{1}{2}-|M|} .
\end{equation*}
In particular,
$|\tilde\varkappa_N| \leq 4$ in Proposition \ref{f_truncated02t}. Furthermore,
\begin{equation*}
\begin{split}
  \lefteqn{\int_0^\infty e^{-wt}\frac{\sinh (Mt)}{\sinh\frac{t}{2}}\, dt}\\
    & = \frac{M}{(w+\frac{1}{2})^2-M^2}
         + \ln\frac{w+\frac{1}{2}+M}{w+\frac{1}{2}-M} \\
    & \quad + 8M(w+\frac{1}{2})
          \int_0^\infty \frac{y}{[ (w+\frac{1}{2}+M)^2+y^2][(w+\frac{1}{2}-M)^2+y^2]}
           \frac{1}{e^{2\pi y}-1}\, dy .
\end{split}
\end{equation*}
In particular, this yields the asymptotics for all $N\in\N$
\begin{equation*}
  \varkappa_N=\int_0^\infty e^{-(N+\frac{1}{2})t}\frac{\sinh ((N+\frac{1}{2})t)}{\sinh\frac{t}{2}}\, dt
    = \ln(4N+3) + c_N, \
     0\leq c_N\leq 2.
\end{equation*}
\end{lemma}
\begin{proof}
For the first inequality one estimates $\cosh(t/2)$ by the exponential
function. For the second inequality, we write for $t>0$
\begin{equation*}
  \frac{1}{\sinh\frac{t}{2}} = 2 e^{-\frac{t}{2}}\frac{1}{1-e^{-t}} 
                    = 2 e^{-\frac{t}{2}}\sum_{j=0}^\infty e^{-jt},
\end{equation*}
and integrate termwise which yields
\begin{equation*}
  \int_0^\infty e^{-wt}\frac{\sinh (Mt)}{\sinh\frac{t}{2}}\, dt
     = 2M\sum_{j=0}^\infty \frac{1}{(w+\frac{1}{2}+j)^2-M^2}.
\end{equation*}
We apply the Abel-Plana summation formula \cite[Th. 4.9c]{Henrici1974}
\begin{equation*}
\begin{split}
  \lefteqn{\frac{1}{2M}\int_0^\infty e^{-wt}\frac{\sinh (Mt)}{\sinh\frac{t}{2}}\, dt}\\
      & = \frac{1}{2}\frac{1}{(w+\frac{1}{2})^2-M^2}
         + \int_0^\infty \frac{1}{(w+\frac{1}{2}+x)^2-M^2}\, dx \\
      & \quad + i\int_0^\infty\Big[ 
            \frac{1}{(w+\frac{1}{2}+iy)^2-M^2} - \frac{1}{(w+\frac{1}{2}-iy)^2-M^2} \Big]
          \frac{1}{e^{2\pi y}-1}\, dy
\end{split}
\end{equation*}
which implies the formula. Finally,
\begin{equation*}
  0\leq \frac{y}{e^{2\pi y}-1} \leq \frac{1}{2\pi}
\end{equation*}
yields the estimate.
\end{proof}